\documentclass[11pt]{article}
\usepackage{booktabs} 
\usepackage[ruled]{algorithm2e} 

\SetAlFnt{\small}
\SetAlCapFnt{\small}
\SetAlCapNameFnt{\small}
\SetAlCapHSkip{0pt}
\IncMargin{-\parindent}
\usepackage{geometry}

\usepackage{vasilis}
\usepackage{balance} 

\newcommand{\opt}{{\rm OPT}}
\newcommand{\spi}{SPI\xspace}
\newcommand{\sps}{SPS\xspace}
\newcommand{\mcg}{MCG\xspace}

\usepackage{thm-restate}

\allowdisplaybreaks

\title{On Submodular Prophet Inequalities and Correlation Gap
  \thanks{Dept.\ of Computer Science, University of
    Illinois, Urbana, IL 61801. {\tt \{chekuri,livanos3\}@illinois.edu}.
    Work on this paper partially supported by NSF grant CCF-1910149.}
  }

\author{Chandra Chekuri
\and
Vasilis Livanos%
}

\date{\today}

\begin{document}

	\maketitle

\begin{abstract}
Prophet inequalities and secretary problems have been extensively
studied in recent years due to their elegance, connections to online
algorithms, stochastic optimization, and mechanism design problems in
game theoretic settings. Rubinstein and Singla \cite{rs} developed a
notion of \emph{combinatorial} prophet inequalities in order to
generalize the standard prophet inequality setting to combinatorial
valuation functions such as submodular and subadditive functions. For
non-negative submodular functions they demonstrated a constant factor
prophet inequality for matroid constraints. Along the way they showed
a variant of the correlation gap for non-negative submodular
functions.

In this paper we revisit their notion of correlation gap as
well as the standard notion of correlation gap and prove much tighter
and cleaner bounds.  Via these bounds and other insights we obtain
substantially improved constant factor combinatorial prophet
inequalities for both monotone and non-monotone submodular functions
over any constraint that admits an Online Contention Resolution
Scheme. In addition to improved bounds we describe efficient
polynomial-time algorithms that achieve these bounds.
\end{abstract}


\thispagestyle{empty}

\pagenumbering{Alph}
\newpage

\pagenumbering{arabic}


\section{Introduction}\label{sec:intro}
Prophet inequalities arose from stochastic optimization and stopping
theory in the '70s.  In the basic setting there are $n$ independent
real-valued random variables $X_1,X_2,\ldots,X_n$ with prescribed
distributions $\cD_1,\ldots,\cD_n$; they correspond to values of some
items. An online algorithm (or agent) knows the distributions of the
random variables a priori but sees their realizations in an
\emph{adversarial} order, and has to choose exactly one of them. The
algorithm has to make an irrevocable decision on whether to accept an
item or not when it arrives. In the single item setting the first
accepted item stops the process. The algorithm's performance is
measured with respect to the value of a prophet who gets to see all
the realizations and then picks the variable with the largest
value. The expected value of the prophet is $V^* = \E[\max_i X_i]$. An
online algorithm $\alpha$-competitive if its expected value is at
least $\alpha V^*$.  Krengel and Sucheston \cite{kren-such}
showed that $1/2$ is the optimal competitive ratio for the single item
setting. Secretary problems are closely related to prophet
inequalities. In the classical version an online algorithm sees $n$ adversarially chosen
values in a \emph{random order} and has to pick one item irrevocably
and compete with the maximum value. A classical result of Dynkin \cite{dynkin}
shows an optimal competitive ratio of $1/e$.

There has been substantial recent interest in prophet inequalities
and secretary problems in theoretical computer science.  Initial
interest arose from strong connections to online mechanism design and
posted price mechanisms for revenue maximization \cite{haji, ChawlaHMS}
(see \cite{Hartline-survey, Chawla-survey} for surveys on Bayesian
mechanism design). Subsequent work explored several different variants
including prophet inequality and secretary problems in more general
settings. Of particular interest to us is the setting where multiple
variables/items from $X_1,X_2,\ldots,X_n$ can be chosen such that the
chosen items are feasible in some combinatorial constraint family.
Two important examples are choosing $k$ items for some integer $k \ge
1$ \cite{Alaei14} and a further generalization where the items chosen
are independent in a matroid \cite{klein-wein}. These generalizations
had several motivations including algorithmic game theory,
combinatorial optimization, stochastic optimization, and online
algorithms. A rich body of work has emerged with several elegant and
useful results. We refer the reader to surveys on prophet
inequalities by Lucier \cite{Lucier-survey} and Correa et
al.\ \cite{Correa-survey}, a survey by Dinitz on the secretary
problem \cite{Dinitz-survey} and by Gupta and Singla generally on
random-order models \cite{GuptaS-survey}, for several pointers to the
extensive literature on these problems and related topics.

\medskip \noindent {\bf Combinatorial Prophet Inequalities:} Prophet
inequalities and secretary problems were mostly studied with
modular/additive objective functions, by which we mean that the total
value of a subset $S$ of variables is simply the sum of their values.
However, more general combinatorial objective functions have many
useful applications. By a combinatorial objective we mean that the
value of a subset of items from $[n]$ is specified by a set function
$f:2^{[n]} \rightarrow \mathbb{R}$. Prominent examples are
submodular\footnote{A real-valued set function $f:2^\cN \rightarrow \mathbb{R}$ is
  submodular if and only if
  $f(A) + f(B) \ge f(A \cup B) + f(A \cap B)$ for all
  $A,B \subseteq \cN$. } and subadditive\footnote{A real-valued set
  function $f:2^\cN \rightarrow \mathbb{R}$ is subadditive if
  $f(A \cup B) \le f(A) + f(B)$ for all $A, B \subseteq \cN$. A
  non-negative submodular function is subadditive.} set functions, and
their special cases. The secretary problem was studied with submodular
objectives and it was shown that it can be reduced to the modular case
with a constant factor loss \cite{submod-sec}. This motivated
Rubinstein and Singla to define a model of combinatorial prophet
inequalities which is the main object of study in this paper. We
restrict our attention to submodular objectives which form a rich
class and, following \cite{rs}, we refer to this as the Submodular
Prophet Inequality (\spi) problem.

The model defined by Rubinstein and Singla is the following elegant
generalization of the standard prophet inequality problem. The input
consists of $n$ independent random variables $X_1,X_2,\ldots,X_n$.
Unlike the standard prophet inequality where $X_i$ is a real-valued
random variable, in the combinatorial setting, each $X_i$ is a
discrete-valued random variable over a finite set $\cU_i$. Thus
$\cD_i$ is a discrete probability distribution over $\cU_i$.  For
technical reasons one assumes that $\cU_1,\cU_2,\ldots,\cU_n$ are
mutually disjoint. Let $\cU = \bigcup_i \cU_i$. There is a
non-negative submodular function $f:2^{\cU} \rightarrow \mathbb{R}_+$
defined over the ground set $\cU$. As in the standard prophet
inequality setting, the variables arrive in an adversarial order. The
online algorithm has to make irrevocable decisions about accepting the
outcome of a variable after seeing its realization
when it arrives, and its goal is to maximize the value $f(S)$ of the
selected set $S \subseteq \cU$.  What about the constraints?  Recall
that in the standard prophet inequality, the goal is to select a
subset of variables from a feasible collection of sets. Similarly, we
assume that there is a downward-closed family of sets
$\cI \subseteq 2^{[n]}$, and it is required that the chosen variables
belong to $\cI$. We emphasize that $\cI$ is defined over $[n]$ and not
$\cU$.  The prophet, which gets to see all the realizations, can
optimize offline, and one can see that its expected profit is
$\E[\max_{S \in \cI} f(\bigcup_{i \in S} \{X_i\})]$.

We briefly motivate a scenario for the setup above. First, we observe
that the standard prophet inequality with additive functions and
arbitrary downward-closed constraints can be modeled by the
combinatorial setting. We simply need to approximate a real-valued
random variable $X_i$, sufficiently closely, by a discrete
distribution over point values. This is relatively easy to do for most
distributions of interest. Now consider the modeling power of the
combinatorial setting. Suppose we have an online advertising situation
where one sees on each day (or time slot) a customer of some type
drawn from a known distribution. It is natural to assume that customer
types are discrete (or can be approximated fairly well by a discrete
distribution with a sufficiently large domain). The agency has to
irrevocably decide whether to show an ad to the customer when they
arrive. The agency has various constraints on the ads that it can
show. For instance, there could be a budget constraint which dictates
that at most $k$ ads can be shown overall which corresponds to picking
$k$ days/variables from the set of $n$ arrivals. What is the value of
serving the ads? That depends on the application at hand. Rich
objective functions allow substantial flexibility in the ability to
model the profit.  For instance it is easy to imagine a decreasing
marginal utility for showing ads to the same type of customers and
this can be easily captured by submodular functions. The model
proposed by \cite{rs} allows arbitrary submodular functions over the
entire universe $\cU$, which is fairly powerful.  The goal of our
description is to point out the generality of the model as it relates
to the traditional prophet inequality setting which has already seen
many applications.

Rubinstein and Singla showed that one can obtain an $O(1)$-competitive
algorithm for \spi when the constraint on selecting the days is a
matroid constraint. However, they did not explicitly consider the case
of monotone submodular functions and they focused only on a single
matroid constraint. Moreover, the constant they obtained is
very large (in the thousands, although they did not try to optimize
it) and they did not consider or emphasize the computational aspects
of the online algorithm. We note that prophet inequalities in the
standard setting of modular/additive objectives are fairly small.
For instance, the well-known result of Kleinberg and Weinberg
\cite{klein-wein} showed a bound of $1/2$ even for arbitrary matroid
constraints, and it is also known that the bound for a cardinality
constraint with $k$ items is $(1-O(1/\sqrt{k}))$ (hence it tends to
$1$ as $k \rightarrow \infty$) \cite{Alaei14}.

\medskip
\noindent {\bf Motivation and technical challenges:} Our main
motivation is to obtain improved bounds for \spi via a clean framework
that applies to a wide variety of constraints. This question was
explicitly raised by Lucier in his survey on prophet inequalities
\cite{Lucier-survey}. Another motivation for
this paper, related to the goal of obtaining improved bounds, comes
from a technical tool that Rubinstein and Singla relied upon, namely
the notion of the correlation gap. One reason for the large constant
in their prophet inequality comes from the correlation gap they prove
for non-negative submodular functions. In addition to the correlation
gap bound, a second technical challenge comes from the model. On each
day $i$ a single element from $\cU_i$ is realized. Thus the overall
distribution over $\cU$ is correlated. Known rounding schemes for
submodular functions such as (Online) Contention Resolution Schemes
\cite{crs,ocrs,rocrs} rely on independence and one needs to suitably
adapt them when handling general classes of constraints beyond a single
matroid, considered in \cite{rs}. We describe the important notion of
the correlation gap that is of independent interest beyond the prophet
inequality setting.

\smallskip
\noindent \emph{Correlation gap of non-negative submodular
  functions:} The term correlation gap of a set function was
introduced in the work of Agrawal et al.\
\cite{adsy-corgap} and has since found several applications. For any
non-negative real-valued set function $f:2^\cN \rightarrow \mathbb{R}_+$
the correlation gap of $f$ is the worst-case ratio of two continuous
extension of $f$, namely the multilinear extension $F$ and the concave
closure $f^+$ \footnote{For some background on submodular
functions along with several lemmas used throughout the paper see
Appendix~\ref{app:background1}.}.
In probabilistic terms, consider a distribution $\cD$
over $\cN =  [n]$ with marginals given by $\bm{x} \in [0,1]^n$. The multilinear
extension $F(\bm{x})$ measures the expected value of $f$ under the product
distribution with marginals $\bm{x}$. The concave closure $f^+(\bm{x})$ gives
the maximum expected value of $f$ over all distributions with
marginals $\bm{x}$. The ratio $\inf_{\bm{x} \in [0,1]^n}\frac{F(\bm{x})}{f^+(\bm{x})}$
is the correlation gap.  If $f$ is a modular/additive function, it is
easy to see that the correlation gap is $1$. An important result in submodular
optimization is that the correlation gap is at most $(1-1/e)$ for any
\emph{monotone}\footnote{A real-valued set function $f$ is monotone if
  $f(A) \le f(B)$ whenever $A \subseteq B$.} submodular function
\cite{adsy-corgap,ccpv,vondrak}. However, it is known that the
correlation gap for general non-negative submodular functions (which
can be non-monotone) can be arbitrarily small. Rubinstein and Singla
\cite{rs} instead use the correlation gap of a related function,
namely, $f_{\max}$ defined as follows:
$f_{\max}(S) = \max_{T \subseteq S} f(T)$. The function $f_{\max}$ is
monotone, but is in general not submodular, even when $f$ is.
It is shown in \cite{rs} that for any non-negative submodular function
$f$, $\inf_{\bm{x} \in [0,1]^n} \frac{F_{\max}(\bm{x})}{f^+(\bm{x})} \ge 1/200$,
where $F_{\max}$ is the multilinear extension of $f_{\max}$.
The main technical claim they prove is that $F\prn{\frac{1}{2} \bm{x}} \ge \frac{1}{200}
f^+(\bm{x})$ for any $\bm{x} \in [0,1]^n$. The proof in \cite{rs} relies on
existing tools but is involved and goes via another continuous extension.
In this paper, we seek to improve the bound but also to give
a refined analysis of the correlation gap for non-negative functions
via a parameter $p = \max_{i \in \cN} x_i$.

\subsection{Our contributions}

In this paper we make two high-level contributions. First, we consider
the correlation gap for non-negative submodular functions both in the
original definition and in the modified sense of \cite{rs} described
above. In both cases, we obtain substantially improved bounds. Second,
we revisit the \spi problem and address three issues: (i)
significantly improved constants for the prophet inequalities for
monotone and non-monotone functions, (ii) a clean black-box reduction
to greedy Online Contention Resolution Schemes that allows one to obtain
prophet inequalities for various other constraints beyond a single
matroid constraint and (iii) computational aspects of the prophet
inequality that were not explicitly addressed in \cite{rs}.
In essence, we answer the open question in \cite{Lucier-survey}
in the affirmative.

\paragraph{Correlation gap:} For a non-negative submodular function,
for any given $p \in [0,1]$, there is a simple instance even when
$n = 2$ such that $F(\bm{x}) \le (1-p)f^+(\bm{x})$ and this implies
that, as $p$ tends to $1$, the correlation gap tends to $0$. In
particular, it turns out that $p = \max_i x_i = |\bm{x}|_{\infty}$.
This phenomenon manifests itself in non-monotone submodular function
maximization in various ways and the typical way to overcome this is
to restrict attention to settings where $p$ is away from $1$.
Nevertheless, there has been little work on precisely quantifying the
correlation gap as a function of this parameter. Our first theorem
addresses this.

\begin{theorem}
  \label{thm:intro-gap}
  Let $f:2^\cN \rightarrow \Rp$ be a non-negative submodular
  function and let $\bm{x} \in [0,1]^n$, where $n = |\cN|$. Let $p = \max_i x_i$. Then
  $F(\bm{x}) \ge (1-p)(1-1/e)f^+(\bm{x})$.  Given any $p \in [0,1]$ there are
  instances such that $F(\bm{x}) \le (1-e^{-(1-p)})f^+(\bm{x})$.
\end{theorem}

The upper
bound of $(1-p)(1-1/e)$ is optimal when $p$ is close to $0$ or when
$p$ is close to $1$. The lower bound on the gap that we show,
$1-e^{-(1-p)}$, agrees nicely with the extremes, but we do not know
whether it is the right bound for all ranges of $p$ and leave it as
an interesting open problem.

We then consider the correlation gap considered by \cite{rs} with
respect to the function $f_{\max}$. We observe that known results on
the Measured Continuous Greedy (\mcg) algorithm \cite{mcg,rocrs}
show that  $F_{\max}(\bm{x}) \ge \frac{1}{e}f^+(\bm{x})$.

We strengthen this observation by considering the parameter $p$.

\begin{theorem}\label{thm:intro-gap-mcg3}
  Let $f:2^\cN \rightarrow \Rp$ be a non-negative submodular function
  and let $\bm{x} \in [0,1]^n$, where $n = |\cN|$. Let
  $p = \max_i x_i$. There exists a point $\bm{y} \in [0,1]^n$, where
  $\bm{y} \le \bm{x}$ (coordinate wise), such that
  $F(\bm{y}) \ge \max\{\frac{1}{e}, (1-p-\frac{1}{e}(1 + \ln(1-p)))\}
  f^+(\bm{x})$.
\end{theorem}

We obtain the preceding theorem as a corollary of the following more
technical theorem.

\begin{restatable}{theorem}{mcggap}\label{thm:intro-gap-mcg2}
  Let $p \in [0,1)$, $f$ be a non-negative submodular function
    with multilinear extension $F$ and $\cP$ be a downward-closed solvable polytope\footnote{Informally,
    a polytope $\cP$ is solvable if one can efficiently do linear optimization over it. A formal
    definition is given in Appendix~\ref{app:background2}.}
  on $\cN$, such that $\cP \subseteq p \cdot [0,1]^{\cN}$ (that is, if
  $\bm{z} \in \cP$ then $z_i \leq p$ for all $i \in \cN$).  Then, the
  output of the Measured Continuous Greedy (\mcg) algorithm on $F$ and $\cP$ at time $b \in [0,1]$ is a
  vector $\bm{x}(b) \in b \cdot \cP$ such that
\[
F(\bm{x}(b)) \geq \begin{cases}
b \cdot e^{-b} \cdot \max_{\bm{z} \in \cP} {f^+(\bm{z})},                               & \text{for } 0 \leq b \leq \ln{\prn{\frac{1}{1-p}}} \\
\prn{1 - p - e^{-b} \prn{1 + \ln{(1-p)}}} \cdot \max_{\bm{z} \in \cP} {f^+(\bm{z})},    & \text{for } \ln{\prn{\frac{1}{1-p}}} \leq b \leq 1.
\end{cases}
\]
\end{restatable}

Theorems~\ref{thm:intro-gap} and \ref{thm:intro-gap-mcg2} are
useful when one has a situation where $p$ is already small and we will
see later that this can indeed be achieved in some cases, such as in
the \spi problem, via a reduction.

\paragraph{Submodular Prophet Inequality:} For \spi we follow the high
level framework of \cite{rs} via the correlation gap and greedy Online
Contention Resolution Schemes (OCRSs) \footnote{The OCRSs that we will
need in this paper are greedy OCRSs. We will abuse notation and omit
greedy for the most part.} \cite{ocrs} (see Appendix~\ref{app:background2}
for some background and formal definitions). Our main contributions
are several technical improvements and refinements that lead to significantly
improved constants and, we believe, clarity on the parameters that
affect the constants. The competitive ratio that we achieve for a
particular constraint family is dictated by the OCRS available for
that family. Roughly speaking, an OCRS for a constraint family is an
online rounding scheme for a given polyhedral relaxation of the
constraints. The approximation quality of the OCRS is governed by two
parameters $b,c \in [0,1]$ via the notion of $(b,c)$-selectability.
For matroids there is $(b,1-b)$-selectable OCRS for any $b \in [0,1]$,
while there is a $(b,e^{-2b})$-OCRS for matching constraints and a
$\prn{1- \frac{t}{\sqrt{k}}, 1- \exp(\frac{-t^2}{4})}$-OCRS for
the special case of a uniform matroid of rank $k$ (picking at most $k$
elements); see \cite{ocrs}\footnote{For the uniform matroid of rank
  $k$ the OCRS we claim is not in \cite{ocrs} but it is easy to derive
and was explicitly done in an unpublished senior thesis \cite{Kalhan-thesis}.}.

\begin{theorem}[{\bf Informal}]
  \label{thm:sps-informal}
  For the \spi problem with a monotone function $f$ over a constraint
  family with a $(b,c)$-selectable OCRS, there is a
  $c \cdot (e^{-b}-\eps) (1-e^{-b})$-competitive algorithm for any fixed
  $\eps > 0$. For non-negative submodular functions there is a
  $\frac{c}{4} \cdot (e^{-b}-\eps) (1-e^{-b})$-competitive algorithm for
  any fixed $\eps > 0$. These competitive ratios can be achieved by an
  efficient randomized polynomial time algorithm, assuming value oracle access
  to $f$ and efficiency of the corresponding OCRS.
\end{theorem}

So far we have avoided mentioning the power of the adversary in
choosing the order of the variables. Our results hold in the setting
of an almighty adversary who can adaptively decide the ordering of the
variables based on the full realization of all the variables and the
choices of the algorithm at each step. We note that the
competitive ratios we obtain are explicit and relatively small. We
summarize our concrete competitive ratios for several constraints of
interest below. OCRSs for constraints can be composed nicely (similar
to CRSs) and our black box reduction is hence very useful.

\begin{table*}[!h]
\begin{small}
\begin{center}
\begin{tabular}{|c|c|c|}
\hline
{Feasibility constraint} & \multicolumn{2}{|c|}{Competitive Ratio} \\
& {Monotone Submodular} & {General Submodular} \\
\hline
\hline
Uniform matroid of rank $k \rightarrow \infty$ & $1/4.3$ & $1/17.2$ \\
\hline
Matroid & $1/7.4$ & $1/30$ \\
\hline
Matching & $1/9.5$ & $1/38$ \\
\hline
  Knapsack & $1/17.5$ & $1/70$ \\
  \hline
  Intersection of $k$ matroids & $\Omega(1/k)$ & $\Omega(1/k)$\\
\hline
\end{tabular}
\caption{A summary of our results for several feasibility constraints.}
\label{tab:concrete-results}
\end{center}
\end{small}
\end{table*}

\subsection{Brief overview of technical ideas}\label{sec:brief-overview}
\emph{Correlation gap:} The correlation gap for monotone functions
\cite{ccpv07,vondrak} used a continuous time argument by relating
$F(\bm{x})$ and $f^+(\bm{x})$ via another continuous extension $f^*$
and this was the same approach followed in \cite{rs}. We take a
different approach. For the exact correlation gap in
Theorem~\ref{thm:intro-gap} we build on a proof for the monotone case
from \cite{crs} which is less well-known; we adapt their proof for the
non-negative case via the parameter $p$. As we remarked, another
approach to bound the correlation gap is via properties of the
\mcg algorithm. The original papers on
the \mcg algorithm \cite{ccpv,mcg} related the quality of
the output to that of the multilinear relaxation. It was only observed
later, motivated by stochastic optimization, that the guarantees of
these algorithms are stronger and can be shown with respect to the
optimum concave closure (see \cite{ward,rocrs}). We build on these
to derive Theorem~\ref{thm:intro-gap-mcg3}.

\smallskip \noindent \emph{\spi:} We follow the high-level approach of
\cite{rs} but make several technical improvements that we briefly
describe. The approach in \cite{rs} is to obtain a fractional solution
offline and then round it online via a greedy OCRS. There are three
technical ingredients.  First is the correlation gap; we already
discussed our improvements on this issue. Second is the fact that the
distribution over $\cU_i$ on each day is a \emph{correlated}
distribution. Rubinstein and Singla essentially show that the limited
correlation in the distribution on each day can be approximately
simulated by a product distribution over $\cU_i$. In this process one
loses large constant factors. We show a simple reduction that reduces
the original \spi instance to another one in which the probability of
any item in $\cU_i$ being realized can be made arbitrarily small. This
allows us to substantially improve the constants in several
interrelated ways.  Third, \cite{rs} rely on the greedy OCRS schemes
from \cite{ocrs} in order to round the fractional solution. As we
remarked, the constraints are on the days/variables while the
objective is defined over $\cU$. In \cite{rs} the authors implicitly
use the fact that the derived constraint on $\cU$ is still a matroid
constraint, if the original constraint on $\cN$ is a matroid. However,
this does not generalize to other constraints. One way to handle this
is to obtain a new OCRS for the derived constraint over $\cU$ from that
on $\cN$. This would lead to technical difficulties and also lose further
constants. In this paper we overcome this difficulty by using the greedy
OCRS for $(\cN,\cI)$ in a black box fashion. We obtain a clean algorithm
that works for any constraint on $\cN$ that admits a greedy OCRS. Our
analysis relies on opening up the internal properties of a greedy OCRS
from \cite{ocrs}.

\subsection{Other related work}
We already referred to recent surveys on prophet inequalities and
secretary problems and related models \cite{Lucier-survey, Correa-survey, GuptaS-survey, Dinitz-survey}.
An older survey on prophet inequalities from a stopping theory point
of view is due to Hill and Kertz \cite{hill-kertz-survey}. The work
here is connected to submodular optimization, stochastic optimization,
online algorithms, and mechanism design which have extensive literature.
It is infeasible to describe all the related work; Singla's thesis
\cite{singla-thesis} touches upon several of these themes and has several
pointers. Contention Resolution Schemes (CRSs) have found many applications
since their introduction \cite{crs}; in fact Bayesian mechanism design
and posted price mechanisms \cite{ChawlaHMS} and subesequent work
by Yan \cite{yan}, connecting mechanism design with the correlation gap,
played an important role in \cite{crs}. Online CRSs were developed
\cite{ocrs} with applications to Bayesian mechanism design as one of
the main motivations and they yield prophet inequalities in the
modular case. Random order CRSs were introduced in \cite{rocrs} and
yield improved bounds when the arrival order is random. We discuss a
variant of the \spi problem for this setting in more detail in
Appendix~\ref{app:sps}. Random Order CRSs found several applications
including to \emph{(submodular) stochastic probing} which has been
extensively studied in the past several years \cite{ward, gupta1, gupta2, gupta3, singla, asad-nazer}.
There are some high-level connections between submodular stochastic
probing and the \spi model, and it is not too surprising that
continuous extensions and CRSs play a role in prior work and ours.
However the probing model requires items to be chosen if examined, and
thus differs from the prophet inequality model, in which items are allowed
to be examined before making the decision. As we remarked already, the
\spi model introduces a particular nuance due to correlations which is
absent in prior models.

Submodular functions and constraints that we consider
such as cardinality, matroids, and others
provide generality and computational tractability. It is
possible to go beyond and consider more general objective
functions such as subadditive and monotone XOS functions, as
well as more complex and general independence constraints.
In such settings one can ignore computational considerations
and focus on the online competitive ratio or assume access to a
demand oracle (even though a demand oracle may be NP-Hard in general).
We refer to \cite{rubin, rs} for some recent work and pointers. Such
functions have also been considered under the related
model of ``combinatorial auctions'' \cite{Dutting2, DuttingKL, AssadiKS},
in which a seller wants to sell distinct items to buyers that have
combinatorial valuation functions for the items. The seller wishes
to maximize either the social welfare or the revenue. In this model,
Dutting, Feldman, Kesselheim and Lucier \cite{Dutting2} obtained a
$2$ prophet inequality for submodular functions, while Dutting,
Kesselheim and Lucier \cite{DuttingKL} obtained a $O(\log \log m)$
prophet inequality for subadditive functions. For the latter, the
authors also show that achieving a constant factor prophet inequality
for subadditive valuation functions is impossible via their techniques
and requires a different approach.

\paragraph{Organization:} Section~\ref{sec:preliminaries} sets up technical
preliminaries on submodular functions and introduces our notation.
Section~\ref{sec:spi} describes the relaxation of the prophet's objective.
Section~\ref{sec:rounding} describes the algorithm and analysis for \spi.
We describe some open questions in Section~\ref{sec:conclusions}.

For several lemmas on submodular functions and sampling used throughout
the paper, see Appendix~\ref{app:background1}. For background information
on contention resolution schemes, see Appendix~\ref{app:background2}.
The proofs of Theorems~\ref{thm:intro-gap} and \ref{thm:intro-gap-mcg2}
have been moved to Appendix~\ref{app:correlation-gap} due to space
constraints. The reduction to small probabilities can be found in
Appendix~\ref{app:reduction}. Appendix~\ref{app:sps} contains a discussion
on a variant of the \spi problem, the \sps problem, in which the random
variables arrive uniformly at random instead of in adversarial order.

\section{Preliminaries}\label{sec:preliminaries}

Let $\cN$ be a finite ground set. A real-valued set function
$f : \{0, 1\}^\cN \to \Rp$ is called {\em submodular} if, for all
$A, B \subseteq \cN$, it satisfies
$f(A) + f(B) \geq f(A \cup B) + f(A \cap B)$. $f$ is {\em
  monotone} if $f(A) \leq f(B)$ for all $A \subseteq B$.
In the rest of this paper we work with non-negative normalized functions that
satisfy $f(\emptyset) = 0$ and $f(A) \geq 0$
for all $A \subseteq \cN$. We often equate $\cN$ with $[n] =
\{1,2,\ldots,n\}$. We use the terminology $S+i$ and $S-i$
as shorthands for $S \cup \{i\}$ and $S \setminus \{i\}$  respectively.
The following continuous extensions of submodular functions to
$[0,1]^\cN$ play an important role in our discussion.

\begin{definition}[Multilinear Extension]
  Let $f : {\{0, 1\}}^\cN \to \R_{\geq 0}$. For any
  $\bm{x} \in {[0, 1]}^n$, let $S \sim \bm{x}$ denote a random set $S$
  that contains each element $i \in \cN$ independently w.p. $x_i$. The
  {\em multilinear extension} of $f$ is defined as
\[
  F(\bm{x}) \coloneqq \E_{S \sim \bm{x}}[f(S)] = \sum_{S \subseteq N}
  {f(S) \prod_{i \in S} {x_i} \prod_{i \notin S} {(1 - x_i)}}.
\]
\end{definition}

\begin{definition}[Concave Closure]
  Let $f : {\{0, 1\}}^\cN \to \R_{\geq 0}$. Moreover, let $\1_S$ denote
  the characteristic vector of length $n = |\cN|$. For any $\bm{x} \in {[0, 1]}^n$,
  the {\em concave closure} of $f$ is defined as
\[
f^+(\bm{x}) \coloneqq \max_{\bm{a} \in [0,1]^{2^\cN}}\left\{\sum_{S \subseteq \cN} {a_S
    f(S)} \midd \sum_{S \subseteq \cN} {a_S} = 1, \sum_{S \subseteq \cN} {a_S \1_S} = \bm{x} \right\}.
\]
\end{definition}

Recall that $f^+(\bm{x})$ can be interpreted as the maximum expected
value of $f(R)$ where $R$ is generated by a distribution whose
marginal values are given by  $\bm{x}$. Since $F(\bm{x})$ corresponds
to the product distribution defined by $\bm{x}$, which is a specific
distribution, it follows that $F(\bm{x}) \le f^+(\bm{x})$
for all $\bm{x}$.

We also need the following notation, which will be useful in our analysis
when dealing with the input constraints.

\begin{definition}[Blowup of a Ground Set]
  Let $\cN$ denote a finite set. Suppose for each $e \in \cN$
  there is an associated finite non-empty set $A_e$
  such that the sets $A_e, e \in \cN$ are mutually disjoint. Let
  $\mathcal{A} = \{A_e \mid e \in N\}$. We call
  $\cN' = \bigcup_{e \in N} A_e$ the {\em blowup} of $\cN$ by
  $\mathcal{A}$.
\end{definition}

Recall that each day only one item from $\cU_i$ is realized and
this motivates the following definition of a constraint family.

\begin{definition}[Partition Extension of a Constraint]
  Let $\cC = (\cN, \cI)$ be a downward-closed constraint family over
  $\cN$. Consider a blowup $\cN'$ of $\cN$ induced by sets
  $A_e, e \in \cN$. Consider the function $g: \cN' \rightarrow \cN$
  where $g(e') =  e$ if and only if $e' \in A_e$. The {\em
    partition extension} of $\cC$, denoted by $\cC'$, is a
  constraint family $(\cN', \cI')$ where
  \[
    \cI'_{A} = \set{S \subseteq \cN' \midd g(S) \in \cI \text{ and }
      \forall e \in \cN, \: |A_e \cap S| \le 1}.
  \]
\end{definition}

\section{Submodular Prophet Inequality Problem}\label{sec:spi}

In the {\em Submodular Prophet Inequality (SPI)} problem, we are given
$n$ random variables $X_1, \dots, X_n$ following (known) distributions
$D_1, \dots, D_n$, along with a constraint $\cC$ on
$\cN = \{1, 2, \dots, n\}$. The random variables are arranged in
adversarial (worst-case) order. Let $\cU_i$ denote the image (range) of
$X_i$, and $\cI$ denote the independent sets of $\cC$.

The online algorithm starts with a set $S = \emptyset$ of selected
elements and a set $Z = \emptyset$ of selected days. At the $i$-th
time step, it is presented with the realization $e \in \cU_i$ of
$X_i$. At that moment, it has to decide irrevocably whether to include
$e$ in $S$ (and hence $i$ in $Z$) or not, subject to $Z$ remaining
independent in $\cC$. The algorithm is also given a non-negative
submodular function $f : \cU \to \Rp$, where
$\cU \triangleq \bigcup_{i = 1}^n {\cU_i}$. The algorithm's objective is to maximize
$f(S)$, subject to $Z$ being independent in $\cC$.

In this model, we are comparing against the almighty adversary who
already knows all realizations and can adaptively change the order in
which to reveal the random variables based on the algorithm's actions
so far and also the random coins it uses (if the algorithm is
randomized). The prophet/adversary will select the best possible set
$S^*$ according to the constraints with knowledge of the
realizations. Thus, we compare the expected value of the online algorithm
against the expected value of the prophet, which is
\[
\opt = \E_{\bm{X}}\left[\max_{T \in \cI} f\left(\left\{X_i \midd i \in T\right\}\right)\right].
\]

Later, we will use an OCRS to round the fractional solution we obtain
in this section. Since $f$ is defined over $\cU$ but the constraint
given is over $\cN$, the days, we cannot immediately apply an OCRS for
rounding. To overcome this issue, we view $\cU$ as the blowup of $\cN$
with respect to $\set{\cU_i}_{i = 1}^n$. On each day $i$, only one
element arrives. Therefore, our input constraint $\cC$ is equivalent
to a new constraint $\cC'$ on $\cU$, where we are allowed to pick only
one element from each day. Notice that this is exactly the partition
extension $\cC'=(\cU,\cI')$ of $\cC$.

We also denote $\bm{X} = \set{X_1, X_2, \dots, X_n}$ and
$\bm{D} = \set{D_1, D_2, \dots, D_n}$. For an element $e \in \cU_i$ we let
$D_i(e)$ denote the probability of $e$ being realized; we also use the
notation $\cD(e)$ to denote the probability of $e \in \cU$ when we do
not need to specify the part it belongs to. Note that the elements
within $\cU_i$ are correlated and hence we do not have a product
distribution on $\cU$.

\paragraph{Algorithmic approach:} Following the description in
Section~\ref{sec:intro}, we design an online algorithm following the
general approach of \cite{rs} but with technical differences. First,
we obtain a relaxation of the prophet's objective. Afterwards, to design
an online algorithm, we obtain an \emph{offline} fractional point $\bm{z}$
based on the input, and round it \emph{online} using a greedy OCRS and other
tools. In this section, we describe the relaxation of the prophet's objective
and how to obtain an offline fractional point $\bm{z}$. The process of
rounding $\bm{z}$ online using a greedy OCRS is presented in
Section~\ref{sec:rounding}.

Before we proceed, we describe a simple but technically important
reduction that allows us to obtain improved bounds.

\begin{observation}[Reduction to small probabilities]\label{obs:reduction}
  Let $I = (\cN, \cU, \bm{D}, \bm{X}, f, \cC)$ be an instance of the
  Submodular Prophet Inequality problem. For every fixed $\eps > 0$,
  there is a reduction of $I$ to another instance
  $I' = (\cN, \cU', \bm{D'}, \bm{Y}, g, \cC)$ of the SPI problem such
  that that (i) for all $e \in \cU'$, $\bm{D'}(e) \le \eps$ and (ii)
  there exists an $\alpha$-competitive algorithm for $I$ if and only
  if there exists an $\alpha$-competitive algorithm for $I'$.
\end{observation}

\begin{remark}
  The reduction's simplicity may make the reader wonder why it is
  useful in achieving improved bounds. The reason is a combination of
  the model and the power of submodularity. The fact that we can only pick a single random
  variable from each day allows us to make copies of the elements, and
  we can use a derived submodular function to treat the
  copies as a single element.
\end{remark}

We describe the reduction in Appendix~\ref{app:reduction}, but only sketch
its correctness since it is rather simple and easy to see, though tedious
to formally prove. The reduction to small probabilities allows us to use
improved correlation gaps, as well as obtain better bounds in the rounding
algorithm.

\subsection{An upper bound on the prophet's value:}

Let $\cP$ denote a solvable polyhedral relaxation of $\cC$ \footnote{For
some background on polyhedral relaxations see Appendix~\ref{app:background2}}.
Then one can easily develop a solvable polyhedral relaxation of $\cC'$ as
follows:
\[
  \cP' = \set{\bm{y} \in [0,1]^{|\cU|} \midd \sum_{e \in \cU_i} y_e = x_i, \,\,\, i \in
  [n], \bm{x} \in \cP}.
\]

Consider any algorithm, including an offline algorithm, that computes a
feasible output given the realizations of the random variables.
For any fixed algorithm $\cA$ (deterministic or
randomized) we have a probability $p_{\cA}(e)$ for each $e \in \cU$
appearing in the output of $\cA$. Since an element $e \in \cU$ is
realized with probability $\cD(e)$, $e$ cannot appear in the output
of $\cA$ with probability more than $\cD(e)$. Moreover, for a given
realization, each output of the algorithm is feasible. Putting these
facts together we obtain the following observation.

\begin{observation}
  Let $\cA$ be any online or offline algorithm for a given instance of
  the  problem. Let $p_{\cA}(e)$ denote the probability that $e$ is in the
  output of $\cA$. Then the vector $\bm{p}$ is in the polytope
\[
\cP'' = \set{ \bm{z} \in [0,1]^{|\cU|} \midd \bm{z} \in \cP', z_e \le
  \cD(e) \,\, e \in \cU}.
\]
\end{observation}

We are now ready to proceed with the relaxation of the prophet's
objective.

\begin{claim}\label{clm:fmax-opt}
Consider an instance of the Submodular Prophet Inequality problem. Then
\[
\max_{\bm{z} \in \cP''} {f^+(\bm{z})} \ge \opt.
\]
\end{claim}
\begin{proof}
  Fix an optimal strategy for the prophet and let
  $\bm{y^*} \in [0,1]^{|\cU|}$ denote the vector of probabilities of the
  elements appearing in the output of the prophet's strategy.  We have
  $\bm{y^*} \in \cP''$. By the definition of the concave closure of
  $f$, $f^+(\bm{y^*})$ maximizes the value of $f$ among all
  distributions with the marginals $\bm{y^*}$ (note that the
  distribution that achieves this may not be a feasible strategy for
  any algorithm). Therefore, $f^+(\bm{y^*}) \geq \opt$, which also implies that
  $\max_{\bm{z} \in \cP''} {f^+(\bm{z})} \geq \opt$.
\end{proof}

\subsection{Fractional Solution and Correlation Gap}\label{sec:fractional}

From Claim~\ref{clm:fmax-opt},
$\max_{\bm{z} \in \cP''} {f^+(\bm{z})} \ge \opt$.  Since OCRSs are designed to relate the
quality of their output to that of the multilinear relaxation, we need
to relate $F(\bm{z})$ to $f^+(\bm{z})$ and hence to $\opt$. We present
two different ways to do this --- via a direct correlation gap and via the Measured
Continuous Greedy (\mcg) algorithm --- with the second yielding strictly
better results than the first.

\paragraph{The direct correlation gap approach}
The first approach is not computationally efficient and relies on
optimally solving the optimization problem
$\max_{\bm{z} \in \cP''} {f^+(\bm{z})}$. Let $\bm{z^*}$ be the optimum
solution. We can then use the correlation gap to relate $F(\bm{z^*})$
to $\opt$.  For monotone functions we have
$F(\bm{z^*}) \ge (1-1/e)f^+(\bm{z^*}) \ge (1-1/e)\opt$. For
non-negative functions we can use Theorem \ref{thm:intro-gap}, the
proof of which, along with all results on the direct correlation gap
approach, can be found in Appendix~\ref{app:correlation-gap}.
Following the reduction that we described earlier, we can assume that
$z_e^* \le \max_e \bm{D}(e) \le \eps$ for all $e$ and this implies,
via Theorem \ref{thm:intro-gap} that
$F(\bm{z^*}) \ge (1-\eps)(1-1/e) f^+(\bm{z^*}) \ge (1-\eps)(1-1/e)
\opt$.  In rounding it is useful to have a solution
$\bm{z} \in b \cdot \cP''$ for some parameter $b \in (0,1)$. One can
of course use $\bm{z} = b\bm{z^*}$ and in this case, we can use the
concavity of $f^+$ to see that $f^+(b\bm{z^*}) \ge b f^+(\bm{z^*})$,
and then apply the correlation gap to $b\bm{z^*}$ to conclude that, in
the monotone case,
$F(b \bm{z^*}) \ge b(1-1/e) f^+(\bm{z^*}) \ge b(1-1/e) \opt$ and, in
the non-monotone case,
$F(b \bm{z^*}) \ge b(1-\eps)(1-1/e) f^+(\bm{z^*}) \ge b(1-\eps)(1-1/e)
\opt$.

\paragraph{The measured continuous greedy approach}

The second approach is algorithmic and relies on the Measured
Continuous Greedy (\mcg) algorithm  and its properties. We state two
relevant known results about the algorithm. For these results as well
as Theorem~\ref{thm:intro-gap-mcg2}, we assume the submodular function $f$
is given via a value oracle, and that the algorithms are randomized and
run in polynomial time and are correct with high probability.

\begin{lemma}[Lemma 4 of \cite{ward}]\label{lem:mon-cont-greedy}
Let $f$ be a monotone submodular function with multilinear extension
$F$, and let $\cP$ be a solvable downward-closed polytope. Let $\bm{x}(b)$ be
solution produced by the Continuous Greedy algorithm on $F$ and $\cP$
until time $b \in (0, 1]$. Then (i) $\bm{x}(b) \in b \cdot \cP$ and (ii)
$F(\bm{x}(b)) \geq \prn{1 - e^{-b} - o(1)} \cdot \max_{\bm{y} \in \cP} {f^+(\bm{y})}$.
\end{lemma}

For a general non-negative submodular function, the \mcg algorithm achieves the following bound.

\begin{lemma}[Lemma 8.3 of \cite{rocrs}]\label{lem:non-mon-cont-greedy}
  Let $b \in [0,1]$, $f$ be a non-negative submodular function
  with multilinear extension $F$, and let $\cP$ be a solvable downward-closed
  polytope. Then, the solution $\bm{x}(b) \in {[0,1]}^n$ produced by
  the \mcg algorithm satisfies (i) $\bm{x}(b)
  \in b \cdot \cP$
  and (ii) $F(\bm{x}(b)) \geq \prn{b \cdot e^{-b} - \eps} \cdot \max_{\bm{y} \in \cP} {f^+(\bm{y})}$,
  for any fixed $\eps > 0$.
\end{lemma}

The two preceding lemmas are algorithmic. If $\cP$ is solvable then
the underlying algorithms can be implemented efficiently. Based on our
reduction to small probabilities it is useful to consider whether the
preceding lemmas can take advantage of this.  No advantage is possible
in the monotone setting, however, we show below that one can indeed
take advantage of the reduction when $f$ is non-monotone. We provide a
refined analysis of the standard bound of the \mcg algorithm, which
depends on a parameter $p$ that quantifies the maximum value of
any coordinate that is feasible in the polytope.
For small enough $p$, Theorem~\ref{thm:intro-gap-mcg2} constitutes an
improvement over Lemma~\ref{lem:non-mon-cont-greedy}, which comprises
the main result of this section. We restate the theorem here, while its
proof can be found in Appendix \ref{app:correlation-gap}. Notice that
Theorem~\ref{thm:intro-gap-mcg3} follows from Theorem~\ref{thm:intro-gap-mcg2}
by setting $b = 1$.

\mcggap*

\begin{remark}\label{rmk:refined-mcg}
Notice that, for the \spi problem, due to our reduction, we can assume
that all vectors $\bm{z} \in \cP''$ have $z_i \leq \eps'$ for all $i \in \cN$,
for any fixed constant $\eps' > 0$. Therefore, for any fixed constant
$\eps > 0$, there exists an $\eps'$ such that
\[
F(\bm{x}(b)) \geq \prn{1 - e^{-b} - \eps} \cdot \max_{\bm{z} \in \cP} {f^+(\bm{z})},
\]
where $\bm{x}(b) \in b \cdot \cP''$ is the output of the \mcg
algorithm at time $b$.
\end{remark}

We summarize the results via both methods below. We observe that for
both monotone and non-monotone functions the bounds are best when
$p \rightarrow 0$, which we can ensure via the reduction. Once we make
this assumption, the bounds provided by the correlation gap
approach are essentially $(1-1/e)$ when $b = 1$ which is
optimal. However, these bounds are matched by the Continuous Greedy
approach. When $b < 1$, which will be the case when applying
the rounding schemes, the bound in Lemma~\ref{lem:mon-cont-greedy} and
our new refined bound in Theorem~\ref{thm:intro-gap-mcg2} are superior and
have the further advantage of being computable in polynomial time.

\section{Rounding the fractional solution}
\label{sec:rounding}
In the preceding section we described ways to obtain a vector
$\bm{z} \in b \cdot \cP''$ for some $b \in [0,1]$ such that
$F(\bm{z}) \geq \alpha \cdot OPT$ for various constants $\alpha$
depending on the approach. In this section we show how to round
$\bm{z}$ in an online fashion. We follow the high-level approach of
\cite{rs} but refine it in several ways.
We will use a greedy OCRS for $\cC$ via
the relaxation $\cP$ as a black box. Recall that our rounding needs to
produce a feasible set in $\cC'$ with ground set $\cU$, while the OCRS
is for the constraint on days/variables $\cN$. Moreover the
distribution $\bm{D}$ is \emph{not} a product distribution on $\cU$.
These are the technical challenges that need to be overcome in the
algorithm and analysis. The quality of the output will depend on the
properties of the OCRS for $\cP$. We assume that the greedy OCRS
for $\cP$ is $(b,c)$-selectable, where $c$ is some function of
$b$. This depends on the specific constraint family $\cC$ and the
polyhedral relaxation $\cP$. At the end of the section, we use
known results to derive concrete competitive ratios for several
constraint families of interest. We note that $\bm{z} \in \cP''$,
which also implies that $\bm{z} \in \cP'$. For rounding purposes
we only work with $\cP'$ and $\cP$; $\cP''$ is only necessary to
obtain an upper bound on $\opt$.

We rely on the certain parts of the analysis of OCRS for submodular
function maximization from \cite{ocrs}. In the following, we will use
$\pi$ to denote the mapping function for the OCRS over the ground set
$\cN$ and the polytope $\cP$. Technically the mapping $\pi$ is a
function of $\bm{x} \in \cP$ and should be written as $\pi_{\bm{x}}$
but we omit $\bm{x}$ for notational simplicity. We also note that $\pi$
can be randomized. An important
definition from \cite{ocrs} in the analysis of OCRSs is the
characteristic CRS of a greedy OCRS.

\begin{definition}[Characteristic CRS of an OCRS]\label{def:char-crs}
  The {\em characteristic CRS} $\bar{\pi}$ of a greedy OCRS $\pi$ for
  a polytope $\cP$ is a CRS for the same polytope $\cP$. It is defined
  for an input $x \in \cP$ and a set $A \subseteq \cN$ by
  $\bar{\pi}(A) = \set{e \in A \midd I \cup \{e\} \in \cF_{\pi,
      \bm{x}}, \: \: \forall I \subseteq A, \: I \in \cF_{\pi,
      \bm{x}}}$. Notice that, if $\pi$ is randomized, then $\bar{\pi}$
      is randomized as well.
\end{definition}

We will also need the following known results from \cite{ocrs}.

\begin{observation}[Observation 3.3 of \cite{ocrs}]\label{obs:char-crs}
  For every set $A \subseteq \cN$ and a characteristic CRS $\bar{\pi}$
  of a greedy OCRS $\pi$, the set $\bar{\pi}(A)$ is always a subset of
  the elements selected by $\pi$ when the active elements are the
  elements of $A$.
\end{observation}

\begin{lemma}[Lemma 3.4 of \cite{ocrs}]\label{lem:char-crs}
  The characteristic CRS $\bar{\pi}$ of a $(b, c)$-selectable greedy
  OCRS $\pi$ is $(b, c)$-balanced and monotone.
\end{lemma}

For any $S \subseteq \cU$, we define $S_\downarrow \subseteq \cN$ to
be the projection of $S$ onto $\cN$, i.e.
\[
  S_\downarrow \coloneqq \set{i \in \cN \midd S \cap \cU_i \neq \emptyset}.
\]

Also, for a greedy OCRS $\pi$, we denote the characteristic CRS of
$\pi$ by $\bar{\pi}$. We now define a CRS $\pi'$ for $\cP'$ that we
will need for our analysis later on. We define $\pi'$ using the
characteristic CRS $\bar{\pi}$ of $\pi$ as follows. For any set
$S \subseteq \cU$,
\[
  \pi'(S) \coloneqq \bigcup_{\stack{i \in \bar{\pi}(S_\downarrow)}{|S \cap \cU_i| = 1}} {(S \cap \cU_i)}.
\]

\begin{lemma}\label{lem:new-crs}
  For any $(b, c)$-selectable greedy OCRS $\pi$ for $\cP$ and
  $\bm{z} \in \cP'$, the mapping $\pi'$ is is a $(b, c \cdot \gamma)$-balanced monotone
  CRS $\pi'$ for $\cP'$, where
  $\gamma = \min_{i \in \cN} {\prod_{e \in \cU_i} {(1 -
      z_e)}}$.
\end{lemma}

\begin{proof}
  First, notice that $\pi'$ is a CRS, since $\pi'(S) \subseteq S$ for
  all $S \subseteq \cU$. This follows immediately from the definition
  of $\pi'$ as $S \cap \cU_i \subseteq S$ for all
  $i \in \cN, S \subseteq \cU$.

  Next, we show that $\pi'$ is monotone. Fix an element
  $e \in S_1 \subseteq S_2 \subseteq \cU$, and an instantiation of
  $\cF_{\pi, \bm{x}}$ (this is relevant if the OCRS is randomized). Let
  $e \in \cU_i$ for some $i \in \cN$. Suppose $e \in \pi'(S_2)$. This
  implies that $|S_2 \cap \cU_i| = 1$ and since $S_1 \subseteq S_2$ and
  $e \in S_1$, we have $|S_1 \cap \cU_i| = 1$. Furthermore, we know that
  $i \in \bar{\pi}({S_2}_\downarrow)$. Since $S_1 \subseteq S_2$, it
  follows that ${S_1}_\downarrow \subseteq {S_2}_\downarrow$. By Lemma
  \ref{lem:char-crs}, we know that $\bar{\pi}$ is monotone, and thus,
  since $i \in \bar{\pi}({S_2}_\downarrow)$, it follows that
  $i \in \bar{\pi}({S_1}_\downarrow)$. Therefore, we know that
  $e \in \pi'(S_1)$. Since $e \in \pi'(S_2)$ implies that
  $e \in \pi'(S_1)$, unconditioning over the instantiation of
  $\cF_{\pi, \bm{x}}$ yields
\[
\Pr\brk{e \in \pi'(S_1)} \geq \Pr\brk{e \in \pi'(S_2)}.
\]

We now show that $\pi'$ is $(b, c \cdot \gamma)$-balanced, for
$\gamma = \min_{i \in \cN} {\prod_{e  \in \cU_i} {(1 -
    z_e)}}$. It suffices to show that, for any $e \in \cU$
\[
\Pr_{S \sim R(\bm{z})}\brk{e  \in \pi'(S) \midd e \in S} \geq c \cdot
\gamma.
\]
Notice that, for any realization $S$ of $R(\bm{z})$,
$e  \in \pi'(S)$ if and only if $S \cap \cU_i = \{e\}$ and
$i \in \bar{\pi}(S_\downarrow)$. Thus,
\begin{align}
\Pr_{S \sim R(\bm{z})}\brk{e \in \pi'(S) \midd e  \in S} &= \Pr_{S \sim R(\bm{z})}\brk{S \cap \cU_i = \{e\} \wedge i \in \bar{\pi}(S_\downarrow) \midd e \in S} \nonumber \\
&= \Pr_{S \sim R(\bm{z})}\brk{S \cap \cU_i = \{e\} \midd e \in S} \nonumber \\
& \qquad \cdot \Pr_{S \sim R(\bm{z})}\brk{i \in \bar{\pi}(S_\downarrow) \midd S \cap \cU_i = \{e\}, e \in S} \nonumber \\
&= \Pr_{S \sim R(\bm{z})}\brk{S \cap \cU_i = \{e\} \midd e \in S} \nonumber \\
\label{eq:bc-balanced-1} & \qquad \cdot \Pr_{S \sim R(\bm{z})}\brk{i \in \bar{\pi}(S_\downarrow) \midd S \cap \cU_i = \{e\}},
\end{align}
where the last equality follows from the fact that, if
$S \cap \cU_i = \{e\}$, then $e \in S$. We lower bound each
probability in \eqref{eq:bc-balanced-1} separately, starting from
\begin{align}\label{eq:bc-balanced-2}
\Pr_{S \sim R(\bm{z})}\brk{S \cap \cU_i = \{e\} \midd e \in S} =
  \prod_{e' \neq e, e' \in \cU_i} {(1 - z_{e'})} \geq \prod_{e'  \in
  \cU_i} {(1 - z_{e'})} \geq \gamma.
\end{align}

Also, notice that $\bar{\pi}$ is a CRS over $\cN$ and does not depend
on which $S \cap \cU_i$ led to $i \in S_\downarrow$. Therefore,
\[
\Pr\brk{i \in \bar{\pi}(S_\downarrow) \midd i \in S_\downarrow} = \Pr\brk{i \in \bar{\pi} \midd S \cap \cU_i = T}
\]
for all $T \subseteq \cU_i$ such that $T \neq \emptyset$. Specifically, for $T = \{e\}$,
\begin{equation}\label{eq:bc-balanced-3}
\Pr\brk{i \in \bar{\pi} \midd S \cap \cU_i = \{e\}} = \Pr\brk{i \in \bar{\pi}(S_\downarrow) \midd i \in S_\downarrow} \geq c,
\end{equation}
where the last inequality follows from the fact that fact that $\bar{\pi}$ is $(b,c)$-balanced, by Lemma \ref{lem:char-crs}.

Combining \eqref{eq:bc-balanced-1}, \eqref{eq:bc-balanced-2} and \eqref{eq:bc-balanced-3}, we obtain
\[
\Pr_{S \sim R(\bm{z})}\brk{e  \in \pi'(S) \midd e \in S} \geq c \cdot \gamma.
\]
\end{proof}

\begin{remark}\label{rmk:small-z}
  Notice that via Observation \ref{obs:reduction}, we can assume
  without loss of generality that, for any fixed $\eps' > 0$,
  $z_{e} \leq \eps'$ for all $e \in \cU$. By choosing $\eps'$
  sufficiently small, for any fixed $\eps > 0$ we have
\[
\gamma = \min_{i \in \cN} {\prod_{e \in \cU_i} {(1 - z_e)}} \geq \min_{i \in \cN} {\prn{\prod_{e \in \cU_i} {e^{- z_{e}}}}} - \eps = \min_{i \in \cN} {\prn{e^{- \sum_{e \in \cU_i} {z_e}}}} - \eps \geq e^{-b} - \eps,
\]
where the last inequality follows from the fact that
$\bm{z} \in b \cdot \cP'$. Thus,
$c \cdot \gamma \geq c \cdot (e^{-b} - \eps)$, and we obtain the following as
corollary: For any $(b, c)$-selectable greedy OCRS $\pi$ for $\cP$ and
fixed $\eps > 0$, $\pi'$ defined earlier is a $\prn{b, c\prn{e^{-b} - \eps}}$-balanced
monotone CRS for $\cP'$.
\end{remark}

Now we are ready to describe our online algorithm. We describe and
analyze the algorithms for monotone and non-monotone cases separately,
since there are technical differences. The algorithms are similar
to the one in \cite{rs}, however, the main technical difference is
that we use the OCRS for $\cN$ as a black box; in \cite{rs} the
authors use an OCRS over $\cU$ since they work in the special case of
matroids.

\subsection{Monotone Non-Negative Submodular Functions}

We assume we have already computed a vector $\bm{z} \in b \cdot \cP''$ for
some $b \in [0,1]$ such that $F(\bm{z}) \ge \alpha \cdot \opt$ for some $\alpha$.
Note that the adversary is almighty and can alter the order in which it feeds the
variables to the algorithm based on knowledge of the full realizations
of the variables and the actions of the algorithm so far.

Let $\bm{z_i}$ denote the product distribution on $\cU_i$ defined by
marginals $z_i(e), e \in \cU_i$. We write $R \sim \bm{z_i}$ to
denote a random set $R \subseteq \cU_i$ realized according to this
product distribution, and we denote $z_i(e)$ by $z_e$ when $i$ is clear
from context or irrelevant. Furthermore, let $\bm{x} \in [0,1]^n$ be the
vector where $x_i = \Pr_{R \sim \bm{z_i}}\brk{R \neq \emptyset} =
1 - \prod_{e \in \cU_i} {(1 - z_e)}$, for all $i \in \cN$. We assume that
$\bm{x}$ is the input vector to our OCRS $\pi_{x}$ for $\cP$ and its characteristic
CRS $\bar{\pi}_{\bm{x}}$. To simplify our notation, we denote $\pi_{\bm{x}}$
and $\bar{\pi}_{\bm{x}}$ by $\pi$ and $\bar{\pi}$, respectively.

\begin{algorithm}\underline{{\sc{Monotone Rounding}}{$\prn{\cU, f, \bm{D}, \cC, \pi, \bm{z}}$}}
\\  $T_{\text{ALG}} = \emptyset$
\\  \For{$h \from 1$ \KwTo $n$}{
    	Let $X_i$ be variable that arrives on day $h$
\\  	Let $e  \in \cU_i$ be the realization of $X_i$
\\      With probability $\frac{\Pr_{R \sim \bm{z_i}}\brk{R = \{e\}}}{\cD_i(e)}$, set $T_i \from \{e\}$
\\      Otherwise, set $T_i$ to be a random subset $R$ of $\cU_i$, drawn according to $\bm{z}_i$, conditioned on $|R| \neq 1$
\\      \If{$T_i \neq \emptyset$}{
            Feed $i$ to OCRS $\pi$ for $\cP$
\\          \If{$\pi$ accepts $i$ \emph{and} $T_i = \{e\}$}{
                $T_{\text{ALG}} \from T_{\text{ALG}} \cup \{e\}$}}}
Return $T_{\text{ALG}}$
\caption{Algorithm for Monotone Non-Negative Submodular Functions}
\label{alg:mon}
\end{algorithm}

The online algorithm on the $h$-th day receives a random variable
$X_i$ decided by the almighty adversary, and once $X_i$ is received the
algorithm also sees the realization $e \in \cU_i$ of $X_i$ according to
the distribution $\cD_i$. The online algorithm generates a random set
$T_i \subseteq \cU_i$ \emph{after} seeing the realization $e$. The
idea is that if one does not see the realization $e$ of $X_i$, the
distribution of $T_i$ appears identical to the product distribution
generated by $\bm{z_i}$. Note that, for $S \subseteq \cU_i$,
$\Pr_{R \sim \bm{z_i}}[R = S] = \prod_{e \in S} z_{e} \prod_{e \in
  \cU \setminus S} \prn{1 - z_{e}}$.

\begin{lemma}\label{lem:ti-indep}
  For any $i \in \cN$ and $S \subseteq \cU_i$,
  \[
    \Pr[T_i = S] = \Pr_{R \sim \bm{z_i}}[R = S].
  \]
\end{lemma}
\begin{proof}
  Let $\cE_e$ be the event that $e \in \cU_i$ is the realization of
  $X_i$. Note that $\Pr[\cE_e] = \cD_i(e)$. Consider
  $S \subseteq \cU_i$ such that $|S| \neq 1$.
  We see from the algorithm's description that

\[
    \Pr[T_i = S\mid \cE_e] = \prn{1 - \frac{\Pr_{R}\brk{R\prn{\bm{z}_i} = \{e\}}}{D_i(e)}} \cdot \frac{\Pr_{R}\brk{R\prn{\bm{z}_i} = S}}{1 - \Pr_{R}\brk{|R\prn{\bm{z}_i}| = 1}}.
  \]
  Summing up over all realizations of $X_i$, we have that, for any $S$
  such that $|S| \neq 1$,
  \begin{align*}
    \Pr\brk{T_i = S} &= \sum_{e \in \cU_i} {D_i(e) \Pr\brk{T_i = S \midd \cE_e}} \\
&= \sum_{e \in \cU_i} {D_i(e) \prn{1 - \frac{\Pr_{R}\brk{R\prn{\bm{z}_i} = \{e\}}}{D_i(e)}} \cdot \frac{\Pr_{R}\brk{R\prn{\bm{z}_i} = S}}{1 - \Pr_{R}\brk{|R\prn{\bm{z}_i}| = 1}}} \\
&= \frac{\Pr_{R}\brk{R\prn{\bm{z}_i} = S}}{1 -
  \Pr_{R}\brk{|R\prn{\bm{z}_i}| = 1}} \cdot \sum_{e \in \cU_i} {D_i(e) \prn{1 - \frac{\Pr_{R}\brk{R\prn{\bm{z}_i} = \{e\}}}{D_i(e)}}} \\
&= \frac{\Pr_{R}\brk{R\prn{\bm{z}_i} = S}}{1 -
  \Pr_{R}\brk{|R\prn{\bm{z}_i}| = 1}} \cdot \prn{\sum_{e \in \cU_i} {D_i(e)} - \sum_{e \in U_i}{\Pr_{R}\brk{R\prn{\bm{z}_i} = \{e\}}}} \\
&= \frac{\Pr_{R}\brk{R\prn{\bm{z}_i} = S}}{1 -
  \Pr_{R}\brk{|R\prn{\bm{z}_i}| = 1}} \cdot \prn{1 - \sum_{e \in \cU_i}{\Pr_{R}\brk{R\prn{\bm{z}_i} = \{e\}}}} \\
&= \frac{\Pr_{R}\brk{R\prn{\bm{z}_i} = S}}{1 - \Pr_{R}\brk{|R\prn{\bm{z}_i}| = 1}} \cdot \prn{1 - \Pr_{R}\brk{|R\prn{\bm{z}_i}| = 1}} \\
&= \Pr_{R}\brk{R\prn{\bm{z}_i} = S}.
\end{align*}

Next, consider any set $S$ with $|S| = 1$ and, without loss of
generality, assume $S = \{e\}$ for some $e \in \cU_i$. It can be
seen from the algorithm description that $T_i = \{e\}$ if and only if
$e$ is the realization of $X_i$ and the algorithm succeeds in
Line 5 in setting $T_i =  \{e\}$ which happens with probability
$\frac{\Pr_{R \sim \bm{z_i}}\brk{R = \{e\}}}{\cD_i(e)}$. Hence
\[
\Pr[T_i = \{e\}] = \cD_i(e) \cdot \frac{\Pr_{R \sim \bm{z_i}}\brk{R =
    \{e\}}}{\cD_i(e)} = \Pr_{R \sim \bm{z_i}}\brk{R = \{e\}},
\]
as desired.
\end{proof}

We now analyze the expected value of $f(T_{\text{ALG}})$ relying
on the CRS $\pi'$ that we set up (this is inspired by the
use of characteristic CRS in \cite{ocrs}).

\begin{lemma}\label{lem:mon-alg}
  Given a $(b, c)$-selectable greedy OCRS $\pi$ for $\cP$, for any
  $\bm{z} \in b \cdot \cP''$ and fixed $\eps > 0$, Algorithm
  \ref{alg:mon} returns a set $T_{\text{ALG}} \subseteq \cU$ such that
  \[
    \E\brk{f(T_{\text{ALG}})} \geq c \prn{e^{-b} - \eps} \cdot F(\bm{z}).
  \]
\end{lemma}
\begin{proof}
  It is easy to see from the algorithm's description that, for any
  $X_i$, only the actual realization of $X_i$ can be potentially chosen
  to be added to $T_{\text{ALG}}$. Furthermore, the variables chosen by
  the algorithm are feasible in $\cC$, since this is ensured by the
  OCRS.

  Let $T_i$ be the random set generated by the online algorithm for
  variable $X_i$. We see that $T_i$ is independent of $T_{i'}$ for
  $i \neq i'$, due to independence of the realization of the random
  variables $X_1,\ldots,X_n$ and the independence of the coins used in
  the algorithm across days. From Lemma~\ref{lem:ti-indep}, the
  distribution of $T_i$ is according to the product distribution
  $R \sim \bm{z_i}$ over $\cU_i$.  Let $Q = \bigcup_{i=1}^n T_i$. It
  follows that $Q$ is a random set drawn from the product distribution
  induced by $\bm{z}$ over $\cU$. Consider the distribution of the set
  $Q_\downarrow \in \cN$. Because of the product distribution of $Q$
  it can be see that the distribution of $Q_\downarrow$ is a product
  distribution on $\cN$ where $i \in \cN$ appears in $Q_\downarrow$
  with probability $x_i = 1 - \prod_{e \in \cU_i} (1-z_e) \le b$ since
  $z \in b \cdot \cP''$. Note that the algorithm feeds $Q_\downarrow$ to the
  OCRS $\pi$ which is $(b,c)$-selectable. Let $\bar{\pi}$ be the
  characteristic CRS of $\pi$.

  Fix a realization $S$ of $Q$, along with an instantiation of $\cF_{\pi, \bm{x}}$.
  Notice $e \in T_{\text{ALG}} \cap \cU_i$ if and only if
  $|S \cap \cU_i| = \{e\}$ and $i \in \pi(S_\downarrow)$. In
  fact,
  \[
    T_{\text{ALG}} = \bigcup_{\stack{i \in \pi(S_\downarrow)}{|S \cap \cU_i| = 1}} {(S \cap \cU_i)},
  \]
  by the description of Algorithm \ref{alg:mon}. By Observation
  \ref{obs:char-crs}, we have $\bar{\pi}(A) \subseteq \pi(A)$ for any
  $A \subseteq \cN$, and thus $\pi'(S) \subseteq
  T_{\text{ALG}}$. Therefore, by the monotonicity of $f$, we have
  $f\prn{T_{\text{ALG}}} \geq f\prn{\pi'(S)}$, and by unconditioning
\[
  \E\brk{f\prn{T_{\text{ALG}}}} \geq \E\brk{f\prn{\pi'(Q)}}.
\]
Finally, by Lemma \ref{lem:new-crs} and Remark \ref{rmk:small-z}, we have that for any $\bm{z} \in b \cdot \cP''$ and any fixed $\eps > 0$,
\[
\E\brk{f\prn{\pi'(Q)}} \geq c \prn{e^{-b} - \eps} \cdot F(\bm{z}),
\]
which yields
\[
\E\brk{f\prn{T_{\text{ALG}}}} \geq c \prn{e^{-b} - \eps} \cdot F(\bm{z}).
\]
\end{proof}

We are now ready for the main theorem of this section, which follows from
Lemmas \ref{lem:mon-alg} and \ref{lem:mon-cont-greedy}, and Claim \ref{clm:fmax-opt}.

\begin{theorem}\label{thm:mon}
  Let $(\cN, \bm{D}, \cC, f)$ be an instance of the Submodular Prophet
  Inequality model and let $OPT$ denote the prophet's value. Given a
  $(b, c)$-selectable greedy OCRS $\pi$ for $\cP$, for a non-negative
  monotone submodular function $f$, $\bm{z} \in b \cdot \cP''$ and
  fixed $\eps > 0$, Algorithm \ref{alg:mon} returns a set
  $T_{\text{ALG}}$ such that
\[
\E\brk{f(T_{\text{ALG}})} \geq c \prn{e^{-b} - \eps} \prn{1 - e^{-b}} \cdot \opt.
\]
\end{theorem}

Next, we provide concrete results for several constraints, given an
OCRS for these constraints. First, we summarize known greedy OCRSs
for various constraints of interest below.

\begin{lemma}[Theorem 1.1 from \cite{ocrs}]\label{lem:ocrs-items}
There exist:
\vspace{-0.5em}
\begin{itemize}
\setlength\itemsep{0.0em}
	\item For every $b \in [0, 1]$, a $(b, 1 - b)$-selectable deterministic greedy OCRS for matroid polytopes.
	\item For every $b \in [0, 1]$, a $(b, e^{-2b})$-selectable randomized greedy OCRS for matching polytopes.
	\item For every $b \in [0, \frac{1}{2}]$, a $(b, \frac{1 - 2b}{2 - 2b})$-selectable randomized greedy OCRS for the natural relaxation of a knapsack constraint.
\end{itemize}
\end{lemma}

By combining Lemma \ref{lem:ocrs-items} with Theorem \ref{thm:mon},
we obtain the following corollary.

\begin{corollary}\label{cor:mon}
  Let $(\cN, \bm{D}, \cC, f)$ be an instance of the Submodular Prophet
  Inequality model and let $OPT$ denote the prophet's value. For a
  non-negative monotone submodular function $f$ and any fixed
  $\eps > 0$, Algorithm \ref{alg:mon} returns a set $T_{\text{ALG}}$
  such that
\begin{align*}
  \E_{\bm{X}, \cT}[f(T_{\text{ALG}})] &\geq (1 - b) \prn{e^{-b} - \eps} \prn{1 - e^{-b}} \cdot OPT, & \forall b \in [0,1], \text{ if } \cC \text{ is a matroid constraint} \\
  \E_{\bm{X}, \cT}[f(T_{\text{ALG}})] &\geq e^{-2b} \prn{e^{-b} - \eps} \left(1 - e^{-b}\right) \cdot OPT, & \forall b \in [0,1], \text{ if } \cC \text{ is a matching constraint} \\
  \E_{\bm{X}, \cT}[f(T_{\text{ALG}})] &\geq \frac{1 - 2b}{2 - 2b} \prn{e^{-b} - \eps} \left(1 - e^{-b}\right) \cdot OPT, & \forall b \in \left[0,\frac{1}{2}\right], \text{ if } \cC \text{ is a knapsack constraint}
\end{align*}
where $\cT = \{T^1, \dots, T^n\}$ denotes the set of random sets Algorithm \ref{alg:mon} generates.
\end{corollary}

\subsection{Non-Negative Submodular Functions}

Below we describe the algorithm for non-negative functions. It is very
similar to the monotone case except for a minor change in accepting an
element $e$; in the final step, the algorithm tosses an additional
random coin and accepts $e$ with probability $1/2$ (see Line 10 in the
algorithm). This is inspired by the similar idea in \cite{ocrs} in
handling non-monotone functions.

\begin{algorithm}\underline{{\sc{General Rounding}}{$\prn{\cU, f, \bm{D}, \cC, \pi, \bm{z}}$}}
\\	$T_{\text{ALG}} = \emptyset$
\\  \For{$h \from 1$ \KwTo $n$}{
		Let $X_i$ be variable that arrives on day $h$
\\  	Let $e  \in \cU_i$ be the realization of $X_i$
\\      With probability $\frac{\Pr_{R \sim \bm{z_i}}\brk{R = \{e\}}}{\cD_i(e)}$, set $T_i \from \{e\}$
\\      Otherwise, set $T_i$ to be a random subset $R$ of $\cU_i$, drawn according to $\bm{z}_i$, conditioned on $|R| \neq 1$
\\      \If{$T_i \neq \emptyset$}{
            Feed $i$ to OCRS $\pi$ for $\cP$
\\          \If{$\pi$ accepts $i$ \emph{and} $T_i = \{e\}$}{
                With probability $\frac{1}{2}$, $T_{\text{ALG}} \from T_{\text{ALG}} \cup \{e\}$}}}
Return $T_{\text{ALG}}$
\caption{Algorithm for General Non-Negative Submodular Functions}
\label{alg:non-mon}
\end{algorithm}

Notice that Lemmas \ref{lem:new-crs} and \ref{lem:ti-indep} still
hold, as they do not depend on the monotonicity of $f$. We present the
following analogue of Lemma \ref{lem:mon-alg} for general submodular
functions. The proof of the next lemma relies on an argument
similar to that in \cite{ocrs}.

\begin{lemma}\label{lem:non-mon-alg}
  Given a $(b, c)$-selectable greedy OCRS $\pi$ for $\cP$, for any
  $\bm{z} \in b \cdot \cP''$ and fixed $\eps > 0$, Algorithm
  \ref{alg:non-mon} returns a set $T_{\text{ALG}} \subseteq \cU$ such
  that
  \[
    \E\brk{f(T_{\text{ALG}})} \geq \frac{c \prn{e^{-b} - \eps}}{4} \cdot F(\bm{z}).
  \]
\end{lemma}
\begin{proof}
  At every step $i$, Algorithm \ref{alg:non-mon} draws a random set
  $T_i$ according to the product distribution on $\cU_i$ with
  probabilities $\bm{z_i}$, by Lemma \ref{lem:ti-indep}. Let
  $Q = \bigcup_{i \in \cN} {T_i}$. Since the realizations between days
  are independent, $Q$ is a random set that follows the product
  distribution on $\cU$ with probabilities $\bm{z}$. Fix a realization
  $S$ of $Q$ and an instantiation of $\cF_{\pi, \bm{x}}$. Notice that
  $e \in T_{\text{ALG}} \cap \cU_i$ if and only if
  $|S \cap \cU_i| = 1$, $i \in \pi(S_\downarrow)$ and the coin toss of
  Line 10 succeeds. In fact, if we denote
  \[
    W = \bigcup_{\stack{i \in \pi(S_\downarrow)}{|S \cap \cU_i| = 1}} {(S \cap \cU_i)},
  \]
  we have that $\E\brk{f(T_{\text{ALG}})} = \E[f(W(1/2))]$, by the
  description of Algorithm \ref{alg:non-mon}. By Observation
  \ref{obs:char-crs}, we have $\bar{\pi}(A) \subseteq \pi(A)$ for any
  $A \subseteq \cN$, and thus $\pi'(S) \subseteq W$. For ease of
  notation, we denote $\pi'(S)$ by $L$. For our fixed choice of $S$
  and $\cF_{\pi, \bm{x}}$, $L$ is deterministic. Therefore, we can
  think of $W(1/2)$ as obtained by first calculating a set $L(1/2)$ in
  which every element of $L$ appears with probability $1/2$
  independently, and then adding to it a random set
  $\Delta \subseteq \cU \setminus L$. The almighty prophet can control
  the order in which the elements arrive, and thus can make the
  distribution of $\Delta$ depend on $L(1/2)$. However, $\Delta$ is
  guaranteed to contain every element with probability at most $1/2$,
  for every given realization of $L(1/2)$. Thus,
\begin{align*}
  \E\brk{f(W(1/2)) \midd S, \cF_{\pi, \bm{x}}} &= \E\brk{f(L(1/2) \cup \Delta) \midd S, \cF_{\pi, \bm{x}}} \\
                                               &= \sum_{B \subseteq L} {\Pr\brk{L(1/2) = B \midd S, \cF_{\pi, \bm{x}}} \cdot \E\brk{f(B \cup \Delta) \midd S, \cF_{\pi, \bm{x}}}} \\
&\geq \sum_{B \subseteq L} {\Pr\brk{L(1/2) = B \midd S, \cF_{\pi, \bm{x}}} \cdot \frac{\E\brk{f(B) \midd S, \cF_{\pi, \bm{x}}}}{4} } \\
&= \frac{\E\brk{f(L(1/2)) \midd S, \cF_{\pi, \bm{x}}}}{2} \\
&= \frac{\E\brk{f(L) \midd S, \cF_{\pi, \bm{x}}}}{4},
\end{align*}
where the first inequality follows from Lemma \ref{lem:buchbinder} since the function $h_B(T) = h(B \cup T)$ is non-negative and submodular for all $B \subseteq \cU$, and the second inequality follows from Lemma \ref{lem:vondrak2}. Taking an expectation over all possible realizations of $S$ and $\cF_{\pi, \bm{x}}$, we obtain
\[
\E\brk{f(W(1/2))} = \E_{S, \cF_{\pi, \bm{x}}}\brk{\E\brk{f(W(1/2)) \midd S, \cF_{\pi, \bm{x}}}} \geq \E_{S, \cF_{\pi, \bm{x}}}\brk{\frac{\E\brk{f(L) \midd S, \cF_{\pi, \bm{x}}}}{4}} = \frac{\E\brk{f(L)}}{4}.
\]
Finally, by Lemma \ref{lem:new-crs} and Remark \ref{rmk:small-z}, we have that for any $\bm{z} \in b \cdot \cP''$ and any fixed $\eps > 0$,
\[
\frac{\E\brk{f(L)}}{4} \geq \frac{c \prn{e^{-b} - \eps}}{4} \cdot F(\bm{z}),
\]
which implies
\[
\E\brk{f(T_{\text{ALG}})} = \E\brk{f(W(1/2))} \geq \frac{c \prn{e^{-b} - \eps}}{4} \cdot F(\bm{z}).
\]
\end{proof}

We are now ready to proceed with the main result for general submodular functions,
which follows from Lemma \ref{lem:non-mon-alg}, Theorem \ref{thm:intro-gap-mcg2}, and
Claim \ref{clm:fmax-opt}.

\begin{theorem}\label{thm:non-mon}
Let $(\cN, \bm{D}, \cC, f)$ be an instance of the Submodular Prophet Inequality model and let $OPT$ denote the prophet's value. Given a $(b, c)$-selectable greedy OCRS $\pi$ for $\cP$, for a non-negative submodular function $f$, $\bm{z} \in b \cdot \cP''$ and fixed $\eps > 0$, Algorithm \ref{alg:non-mon} returns a set $T_{\text{ALG}}$ such that
\[
\E\brk{f(T_{\text{ALG}})} \geq \frac{c \prn{e^{-b} - \eps}}{4} \cdot \prn{1 - e^{-b} - \eps} \cdot OPT.
\]
\end{theorem}

By combining Lemma \ref{lem:ocrs-items} with Theorem \ref{thm:non-mon},
we obtain the following corollary.

\begin{corollary}\label{cor:non-mon}
Let $(\cN, \bm{D}, \cC, f)$ be an instance of the Submodular Prophet Inequality model and let $OPT$ denote the prophet's value. For a non-negative submodular function $f$ and any fixed $\eps > 0$, Algorithm \ref{alg:non-mon} returns a set $T_{\text{ALG}}$ such that
\begin{align*}
\E[f(T_{\text{ALG}})] &\geq \frac{(1 - b) \prn{e^{-b} - \eps}}{4} \cdot \prn{1 - e^{-b} - \eps} \cdot OPT, & \forall b \in [0,1],  \text{ if } \cC \text{ is a matroid constraint} \\
\E[f(T_{\text{ALG}})] &\geq \frac{e^{-2b} \prn{e^{-b} - \eps}}{4} \cdot \prn{1 - e^{-b} - \eps} \cdot OPT, & \forall b \in [0,1],  \text{ if } \cC \text{ is a matching constraint} \\
\E[f(T_{\text{ALG}})] &\geq \frac{(1 - 2b) \prn{e^{-b} - \eps}}{8 - 8b} \cdot \prn{1 - e^{-b} - \eps} \cdot OPT, & \forall b \in \brk{0, \frac{1}{2}},  \text{ if } \cC \text{ is a knapsack constraint}
\end{align*}
where $\cT = \{T^1, \dots, T^n\}$ denotes the set of random sets Algorithm \ref{alg:mon} generates.
\end{corollary}

\section{Conclusion} \label{sec:conclusions}
We presented a general framework for submodular prophet inequalities
in the model of \cite{rs} via greedy Online Contention Resolution
Schemes and correlation gaps. The framework yields substantially
improved constant factor competitive ratios for both monotone and
general submodular functions, and can be implemented in polynomial time for many
classes of interesting constraints.
The framework resolves the open question posed in \cite{Lucier-survey}
regarding the model of \cite{rs}.

Along the way, we strengthened the notion of correlation gap for non-negative submodular
functions introduced in \cite{rs}, and provided a fine-grained variant of the standard
correlation gap. For both cases, our bounds are cleaner and tighter. Moreover, we presented
a refined analysis of the Measured Continuous Greedy algorithm for polytopes with small coordinates
and general non-negative submodular functions, showing that, for these cases, it yields a bound
that matches the bound of Continuous Greedy for the monotone case.

An interesting open question is whether our fine-grained correlation gap for general non-negative
submodular functions can be made tight. It is tempting to conjecture that the lower bound
on the gap shown in Theorem \ref{thm:lower-gap} is tight for all values of $p$.
We leave this as an interesting open problem to resolve. It is also interesting to obtain
further improvements in the bounds we showed for \spi. Of particular interest is the
cardinality constraint.

\paragraph{Acknowledgements:} 
The authors thank Sahil Singla for
clarifications on \cite{rs}. CC thanks Vondr\'{a}k and VL thanks
Maria Siskaki, Kesav Krishnan and Venkata Sai Bavisetty for
helpful discussions.

\balance
\bibliographystyle{plain}
\bibliography{references}

\begin{thebibliography}{10}

\bibitem{ward}
Marek Adamczyk, Maxim Sviridenko, and Justin Ward.
\newblock Submodular stochastic probing on matroids.
\newblock {\em Mathematics of Operations Research}, 41(3):1022--1038, 2016.

\bibitem{rocrs}
Marek Adamczyk and Michal Wlodarczyk.
\newblock Random order contention resolution schemes.
\newblock In {\em 59th {IEEE} Annual Symposium on Foundations of Computer
  Science, {FOCS} 2018, Paris, France, October 7-9, 2018}, pages 790--801,
  2018.

\bibitem{adsy-corgap}
Shipra Agrawal, Yichuan Ding, Amin Saberi, and Yinyu Ye.
\newblock Price of correlations in stochastic optimization.
\newblock {\em Operations Research}, 60(1):150--162, 2012.
\newblock Preliminary version in Proc.\ of ACM-SIAM SODA 2010.

\bibitem{Alaei14}
Saeed Alaei.
\newblock Bayesian combinatorial auctions: Expanding single buyer mechanisms to
  many buyers.
\newblock {\em SIAM Journal on Computing}, 43(2):930--972, 2014.

\bibitem{asad-nazer}
Arash Asadpour and Hamid Nazerzadeh.
\newblock Maximizing stochastic monotone submodular functions.
\newblock {\em Management Science}, 62(8):2374--2391, 2016.

\bibitem{AssadiKS}
Sepehr Assadi, Thomas Kesselheim, and Sahil Singla.
\newblock Improved truthful mechanisms for subadditive combinatorial auctions:
  Breaking the logarithmic barrier.
\newblock In {\em Proceedings of the 2021 ACM-SIAM Symposium on Discrete
  Algorithms (SODA)}, pages 653--661. SIAM, 2021.

\bibitem{singla}
Domagoj Bradac, Sahil Singla, and Goran Zuzic.
\newblock {(Near) Optimal Adaptivity Gaps for Stochastic Multi-Value Probing}.
\newblock In Dimitris Achlioptas and L{\'a}szl{\'o}~A. V{\'e}gh, editors, {\em
  Approximation, Randomization, and Combinatorial Optimization. Algorithms and
  Techniques (APPROX/RANDOM 2019)}, volume 145 of {\em Leibniz International
  Proceedings in Informatics (LIPIcs)}, pages 49:1--49:21, Dagstuhl, Germany,
  2019. Schloss Dagstuhl--Leibniz-Zentrum fuer Informatik.

\bibitem{bf-survey}
Niv Buchbinder and Moran Feldman.
\newblock Submodular functions maximization problems.
\newblock In Teofilo~F Gonzalez, editor, {\em Handbook of Approximation
  Algorithms and Metaheuristics: Methologies and Traditional Applications,
  Volume 1}. CRC Press, 2nd edition, 2018.

\bibitem{buchbinder}
Niv Buchbinder, Moran Feldman, Joseph~(Seffi) Naor, and Roy Schwartz.
\newblock Submodular maximization with cardinality constraints.
\newblock In {\em Proceedings of the Twenty-fifth Annual ACM-SIAM Symposium on
  Discrete Algorithms}, SODA '14, pages 1433--1452, Philadelphia, PA, USA,
  2014. Society for Industrial and Applied Mathematics.

\bibitem{ccpv07}
G~Calinescu, C~Chekuri, M~Pal, and J~Vondrak.
\newblock Maximizing a submodular function subject to a matroid constraint.
\newblock In {\em the Proceedings of the Conference on Integer Programming and
  Combinatorial Optimization (IPCO)}, 2007.

\bibitem{ccpv}
Gruia Calinescu, Chandra Chekuri, Martin Pál, and Jan Vondrák.
\newblock Maximizing a monotone submodular function subject to a matroid
  constraint.
\newblock {\em SIAM Journal on Computing}, 40(6):1740--1766, 2011.

\bibitem{ChawlaHMS}
Shuchi Chawla, Jason~D Hartline, David~L Malec, and Balasubramanian Sivan.
\newblock Multi-parameter mechanism design and sequential posted pricing.
\newblock In {\em Proceedings of the forty-second ACM symposium on Theory of
  computing}, pages 311--320, 2010.

\bibitem{Chawla-survey}
Shuchi Chawla and Balasubramanian Sivan.
\newblock Bayesian algorithmic mechanism design.
\newblock {\em SIGecom Exch.}, 13(1):5–49, November 2014.

\bibitem{crs}
Chandra Chekuri, Jan Vondr\'{a}k, and Rico Zenklusen.
\newblock Submodular function maximization via the multilinear relaxation and
  contention resolution schemes.
\newblock In {\em Proceedings of the Forty-third Annual ACM Symposium on Theory
  of Computing}, STOC '11, pages 783--792, New York, NY, USA, 2011. ACM.

\bibitem{Correa-survey}
Jos{\'e} Correa, Patricio Foncea, Ruben Hoeksma, Tim Oosterwijk, and Tjark
  Vredeveld.
\newblock Recent developments in prophet inequalities.
\newblock {\em ACM SIGecom Exchanges}, 17(1):61--70, 2019.

\bibitem{Dinitz-survey}
Michael Dinitz.
\newblock Recent advances on the matroid secretary problem.
\newblock {\em ACM SIGACT News}, 44(2):126--142, 2013.

\bibitem{DuttingKL}
Paul D{\"u}tting, Thomas Kesselheim, and Brendan Lucier.
\newblock An $o(\log \log m)$ prophet inequality for subadditive combinatorial
  auctions.
\newblock {\em ACM SIGecom Exchanges}, 18(2):32--37, 2020.

\bibitem{dynkin}
E.~B. Dynkin.
\newblock {The optimum choice of the instant for stopping a Markov process}.
\newblock {\em Soviet Math. Dokl}, 4, 1963.

\bibitem{Dutting2}
Paul Dütting, Michal Feldman, Thomas Kesselheim, and Brendan Lucier.
\newblock Prophet inequalities made easy: Stochastic optimization by pricing
  nonstochastic inputs.
\newblock {\em SIAM Journal on Computing}, 49(3):540--582, 2020.

\bibitem{EhsaniHKS18}
Soheil Ehsani, MohammadTaghi Hajiaghayi, Thomas Kesselheim, and Sahil Singla.
\newblock Prophet secretary for combinatorial auctions and matroids.
\newblock In {\em Proceedings of the Twenty-Ninth Annual ACM-SIAM Symposium on
  Discrete Algorithms}, SODA '18, page 700–714, USA, 2018. Society for
  Industrial and Applied Mathematics.

\bibitem{esf-prophsec}
Hossein Esfandiari, MohammadTaghi Hajiaghayi, Vahid Liaghat, and Morteza
  Monemizadeh.
\newblock Prophet secretary.
\newblock {\em CoRR}, abs/1507.01155, 2015.

\bibitem{mcg}
Moran Feldman, Joseph~(Seffi) Naor, and Roy Schwartz.
\newblock A unified continuous greedy algorithm for submodular maximization.
\newblock In {\em Proceedings of the 2011 IEEE 52Nd Annual Symposium on
  Foundations of Computer Science}, FOCS '11, pages 570--579, Washington, DC,
  USA, 2011. IEEE Computer Society.

\bibitem{ocrs}
Moran Feldman, Ola Svensson, and Rico Zenklusen.
\newblock Online contention resolution schemes.
\newblock In {\em Proceedings of the Twenty-seventh Annual ACM-SIAM Symposium
  on Discrete Algorithms}, SODA '16, pages 1014--1033, Philadelphia, PA, USA,
  2016. Society for Industrial and Applied Mathematics.

\bibitem{submod-sec}
Moran Feldman and Rico Zenklusen.
\newblock The submodular secretary problem goes linear.
\newblock {\em SIAM Journal on Computing}, 47(2):330--366, 2018.

\bibitem{gupta1}
Anupam Gupta and Viswanath Nagarajan.
\newblock A stochastic probing problem with applications.
\newblock In Michel Goemans and Jos{\'e} Correa, editors, {\em Integer
  Programming and Combinatorial Optimization}, pages 205--216, Berlin,
  Heidelberg, 2013. Springer Berlin Heidelberg.

\bibitem{gupta3}
Anupam Gupta, Viswanath Nagarajan, and Sahil Singla.
\newblock {\em Adaptivity Gaps for Stochastic Probing: Submodular and XOS
  Functions}, pages 1688--1702.

\bibitem{gupta2}
Anupam Gupta, Viswanath Nagarajan, and Sahil Singla.
\newblock {\em Algorithms and Adaptivity Gaps for Stochastic Probing}, pages
  1731--1747.

\bibitem{GuptaS-survey}
Anupam Gupta and Sahil Singla.
\newblock Random-order models, 2020.
\newblock \url{https://arxiv.org/abs/2002.12159}.

\bibitem{haji}
Mohammad~Taghi Hajiaghayi, Robert Kleinberg, and Tuomas Sandholm.
\newblock Automated online mechanism design and prophet inequalities.
\newblock In {\em Proceedings of the 22Nd National Conference on Artificial
  Intelligence - Volume 1}, AAAI'07, pages 58--65. AAAI Press, 2007.

\bibitem{Hartline-survey}
Jason~D. Hartline.
\newblock Bayesian mechanism design.
\newblock {\em Foundations and Trends® in Theoretical Computer Science},
  8(3):143--263, 2013.

\bibitem{hill-kertz-survey}
Theodore~P Hill and Robert~P Kertz.
\newblock A survey of prophet inequalities in optimal stopping theory.
\newblock {\em Contemp. Math}, 125:191--207, 1992.

\bibitem{Kalhan-thesis}
Sanchit Kalhan.
\newblock A study on matroid prophet inequalities.
\newblock Senior thesis, University of Illinois at Urbana-Champaign, May 2016.

\bibitem{klein-wein}
Robert Kleinberg and Seth~Matthew Weinberg.
\newblock Matroid prophet inequalities.
\newblock In {\em Proceedings of the forty-fourth annual ACM symposium on
  Theory of computing}, pages 123--136, 2012.

\bibitem{kren-such}
Ulrich Krengel and Louis Sucheston.
\newblock Semiamarts and finite values.
\newblock {\em Bull. Amer. Math. Soc.}, 83(4):745--747, 07 1977.

\bibitem{Lucier-survey}
Brendan Lucier.
\newblock An economic view of prophet inequalities.
\newblock {\em SIGecom Exch.}, 16(1):24–47, September 2017.

\bibitem{rubin}
Aviad Rubinstein.
\newblock {\em Beyond Matroids: Secretary Problem and Prophet Inequality with
  General Constraints}, page 324–332.
\newblock Association for Computing Machinery, New York, NY, USA, 2016.

\bibitem{rs}
Aviad Rubinstein and Sahil Singla.
\newblock Combinatorial prophet inequalities.
\newblock In {\em Proceedings of the Twenty-Eighth Annual ACM-SIAM Symposium on
  Discrete Algorithms}, pages 1671--1687. SIAM, 2017.
\newblock Longer ArXiv version is at \url{http://arxiv.org/abs/1611.00665}.

\bibitem{singla-thesis}
Sahil Singla.
\newblock {\em Combinatorial Optimization Under Uncertainty: Probing and
  Stopping-Time Algorithms}.
\newblock PhD thesis, CMU, 2018.
\newblock
  \url{http://reports-archive.adm.cs.cmu.edu/anon/2018/CMU-CS-18-111.pdf}.

\bibitem{vondrak}
Jan Vondr{\'a}k.
\newblock {\em Submodularity in combinatorial optimization}.
\newblock PhD thesis, Univerzita Karlova, Matematicko-fyzik{\'a}ln{\'\i}
  fakulta, 2007.

\bibitem{yan}
Qiqi Yan.
\newblock Mechanism design via correlation gap.
\newblock In Dana Randall, editor, {\em Proceedings of the Twenty-Second Annual
  {ACM-SIAM} Symposium on Discrete Algorithms, {SODA} 2011, San Francisco,
  California, USA, January 23-25, 2011}, pages 710--719. {SIAM}, 2011.

\end{thebibliography}

\appendix

\section{Background on Submodularity and Contention Resolution Schemes}\label{app:background}

\subsection{Useful Lemmas}\label{app:background1}

Below we state several relevant lemmas regarding sampling and
submodular functions that we need.

\begin{lemma}[Lemma 2.2 from \cite{buchbinder}]\label{lem:buchbinder}
Let $f : 2^\cN \to \Rp$ be submodular. Denote by $A(p)$ a random subset of $A$ where each element
appears with probability at most $p$ (not necessarily independently). Then,
\[
\E[f(A(p))] \geq (1-p) \cdot f(\emptyset).
\]
\end{lemma}

\begin{lemma}[Lemma 2.2 from \cite{vondrak}]\label{lem:vondrak2}
Let $g : 2^\cN \to \Rp$ be submodular. Denote by $A(p)$ a random subset of $A$ where each element appears with probability exactly $p$ (not necessarily independently). Then
\[
\E[g(A(p))] \geq (1 - p) \cdot g(\emptyset) + p \cdot g(A).
\]
\end{lemma}

\begin{lemma}[Lemma 2.3 from \cite{vondrak}]\label{lem:vondrak3}
Let $f : 2^\cN \to \Rp$ be submodular, $A, B \subseteq \cN$ two (not necessarily disjoint) sets and $A(p), B(q)$ their independently sampled subsets, where each element of $A$ appears in $A(p)$ with probability $p$ and each element of $B$ appears in $B(q)$ with probability $q$. Then
\[
\E[f(A(p) \cup B(q))] \geq (1-p)(1-q) \cdot f(\emptyset) + p(1-q) \cdot f(A) + (1-p)q \cdot f(B) + pq \cdot f(A \cup B).
\]
\end{lemma}

\begin{lemma}[Lemma 4.3 from \cite{vondrak}]\label{lem:sampling}
  Let $f : 2^\cN \to \Rp$ be a submodular function, let
  $A_1, A_2, \dots, A_k \subseteq \cN$ be $k$ (not necessarily disjoint)
  sets and let $A_1(p_1), A_2(p_2), \dots, A_k(p_k)$ their
  independently sampled subsets, where each element of $A_i$ appears
  in $A_i(p_i)$ with probability $p_i$, for all $1 \leq i \leq
  k$. Then
  \[
    \E\left[f\left(\bigcup_{i = 1}^k A_i(p_i)\right)\right] \geq \sum_{I \subseteq [k]} {\left(\prod_{j \in I} {p_j} \prod_{j \notin I} {(1 - p_j)} f\left(\bigcup_{j \in I} {A_j} \right)\right)}.
  \]
\end{lemma}

The next Lemma appears in \cite{crs}, but its proof is slightly
obfuscated within Lemma B.2. For clarity, we present it here on its
own.

\begin{lemma}[\cite{crs}]\label{lem:crs}
Let $a_1 \geq \dots \geq a_m \in \Rp$, and $q_1, \dots, q_m \in [0, 1]$ such that $\sum_{k = 1}^m {q_k} = 1$. Then
\[
\sum_{k = 1}^m {q_k \: a_k \prod_{j = 1}^{k-1} {(1 - q_j)}} \geq \left(1 - \frac{1}{e}\right) \cdot \sum_{j = 1}^m {q_j a_j}.
\]
\end{lemma}
\begin{proof}
Since the above inequality is linear in the parameters $a_i$, it suffices to
prove it for the special case $a_1 = a_2 = \dots = a_r = 1$ and
$a_{r+1} = \dots = a_m = 0$. (A general decreasing sequence of $a_j$ can be
obtained as a positive linear combination of such special cases). Hence, it
remains to prove
\[
\sum_{k = 1}^r {q_k \prod_{j = 1}^{k-1} {(1 - q_j)}} \geq \left(1 - \frac{1}{e}\right) \cdot \sum_{j = 1}^r {q_j}.
\]
We start from the left-hand side, which we expand to
\[
\sum_{k = 1}^r {q_k \prod_{j = 1}^{k-1} {(1 - q_j)}} = 1 - \prod_{k = 1}^r {(1 - q_k)} \geq 1 - {\left(1 - \frac{1}{r} \sum_{k = 1}^r {q_k} \right)}^r,
\]
where the inequality follows from the arithmetic-geometric mean inequality. Finally, we use the concavity of $\phi_r(x) \coloneqq 1 - {\left(1 - \frac{x}{r}\right)}^r$, and the fact that $\phi_r(0) = 0$, to get
\[
\phi_r(x) \geq \phi_r(1) \cdot x = \left(1 - {\left(1 - \frac{1}{r}\right)}^r\right) \cdot x
\]
for $x \in [0, 1]$. Since $\left(1 - {\left(1 - \frac{1}{r}\right)}^r\right) \geq \left(1 - \frac{1}{e}\right)$ for all $r$, we get
\[
\phi_r(x) \geq \left(1 - \frac{1}{e}\right) \cdot x.
\]
which implies that
\[
\phi_r\left(\sum_{k = 1}^r {q_k} \right) = 1 - {\left(1 - \frac{1}{r} \sum_{k = 1}^r {q_k} \right)}^r \geq \left(1 - \frac{1}{e}\right) \cdot \sum_{k = 1}^r {q_k}.
\]
\end{proof}

\subsection{Constraints and rounding via Contention Resolution Schemes}\label{app:background2}

Let $\cN$ be a finite ground set. A constraint family over $\cN$ is simply
a subset $\cI \subseteq 2^\cN$; a set $S \in \cI$ is called feasible, while a
set $S \not \in \cI$ is called infeasible. We are interested only in
downward-closed families of constraints; $\cI$ is downward-closed if and only if
$A \in \cI$ implies that any set $B \subset A$ is also in $\cI$.  Classical
examples of downward-closed families include those induced by a matroid on
$\cN$ or intersections of several matroids on $\cN$, independent sets of
graphs, matchings in graphs and hypergraphs, boolean vectors that
satisfy packing constraints of the form $A\bm{x} \le b$ for non-negative
$A,b$, among many others. We will use the terminology $(\cN,\cI)$ to
indicate a constraint family. The maximum weight independent set
problem over $(\cN,\cI)$ is the following: given
$w: \cN \rightarrow \mathbb{R}$ solve $\max_{S \in \cI} w(S)$ where
$w(S) = \sum_{e \in \cN} w(e)$.  Since many of these problems are
NP-Hard, a common technique is to use polyhedral (or more generally
convex) relaxations.  We say $\cP \subseteq [0,1]^\cN$ is a polyhedral
relaxation of $(\cN,\cI)$ if $\cP$ is a polyhedron and $\1_S \in \cP$ for
all $S \in \cI$ (here $\1_S$ is the characteristic vector of $S$).  We
say that $\cP$ is \emph{solvable} if one can efficiently do linear
optimization over $\cP$, that is, given $w : \cN \rightarrow \mathbb{R}$,
there is a polynomial time algorithm that computes
$\max_{\bm{x} \in \cP} \sum_i w_ix_i$. Via the multilinear relaxation,
the polyhedral approach to approximation has been extended successfully
to submodular function maximization \cite{ccpv,crs,bf-survey}.

\paragraph{Contention Resolution Schemes:}
These are rounding schemes introduced in \cite{crs} for submodular
function maximization.

\begin{definition}[Contention Resolution Scheme]
  Let $b, c \in [0,1]$. A $(b, c)$-balanced {\em Contention Resolution
    Scheme} $\pi$ for $\cP_\cI$ is a procedure that for every
  $\bm{x} \in b \cdot \cP_{\cI}$ and $A \subseteq \cN$, returns a
  random set $\pi_{\bm{x}}(A) \subseteq A \cap \text{support}(\bm{x})$
  and satisfies the following properties:
\begin{enumerate}
\item $\pi_{\bm{x}}(A) \in \cI$ with probability
  $1, \quad \forall A \subseteq \cN, \bm{x} \in b \cdot \cP_{\cI}$,
  and
\item for all $i \in \text{support}(\bm{x})$, $\Pr\brk{i \in \pi_{\bm{x}}(R(\bm{x})) \midd i \in R(\bm{x})} \geq c, \quad \forall \bm{x} \in b \cdot \cP_{\cI}$.
\end{enumerate}

The scheme is said to be {\em monotone} if $\Pr\brk{i \in \pi_{\bm{x}}(A_1)} \geq \Pr\brk{i \in \pi_{\bm{x}}(A_2)}$ whenever $i \in A_1
\subseteq A_2$.
\end{definition}

CRSs are offline rounding schemes. {\em Online Contention
Resolution Schemes (OCRS)} were introduced by Feldman, Svensson and
Zenklusen \cite{ocrs} to handle online settings such as the \spi problem
where the arrival order of the elements is adversarial. {\em Random Order
Contention Resolution Schemes (ROCRS)} were introduced by \cite{rocrs}
to handle the cases where the arrival of the elements is a uniformly
random permutation.

\begin{definition}[Online Contention Resolution Scheme (OCRS)]
  Let us consider the following online selection setting. A point
  $\bm{x} \in \cP$ is given and let $R(\bm{x})$ be a random subset of
  active elements. The elements $e \in \cN$ reveal one by one whether
  they are active, i.e., $e \in R(\bm{x})$, and the decision whether
  to select an active element is taken irrevocably before the next
  element is revealed. An {\em Online Contention Resolution Scheme}
  for $\cP$ is an online algorithm that selects a subset
  $I \subseteq R(\bm{x})$ such that $\1_I \in \cP$.
\end{definition}

Monotonicity of a CRS is important for rounding the multilinear
relaxation of submodular functions \cite{crs}, although such a
condition is not needed for modular functions. In the online setting,
\cite{ocrs} defines the notion of a \emph{greedy} OCRS which is
helpful in rounding for submodular functions.

\begin{definition}[Greedy OCRS]
  Let $\cP \subseteq {[0,1]}^n$ be a relaxation for the feasible sets
  $\cF \subseteq 2^\cN$. A {\em greedy OCRS} $\pi$ for $\cP$ is an
  OCRS that for any $\bm{x} \in \cP$ defines a down-closed subfamily
  of feasible sets $\cF_{\pi, \bm{x}} \subseteq \cF$, and an element
  $e$ is selected when it arrives if, together with the already
  selected elements, the obtained set is in $\cF_{\pi, \bm{x}}$. If
  the choice of $\cF_{\pi, \bm{x}}$ given $\bm{x}$ is randomized, we
  talk about a randomized greedy OCRS; otherwise, we talk about a
  deterministic greedy OCRS.
\end{definition}

For a greedy OCRS, the quality of the approximation guaranteed
with respect to the multilinear relaxation is governed by
the notion of $(b,c)$-selectability \cite{ocrs}.

\begin{definition}[$(b,c)$-selectability]
  Let $b, c \in [0,1]$. A greedy OCRS for $\cP$ is $(b, c)$-selectable
  if for any $\bm{x} \in b \cdot \cP$, we have
  \[
    \Pr\brk{I \cup \{e\} \in \cF_{\pi, \bm{x}} \quad \forall I \subseteq R(\bm{x}), \: I \in \cF_{\pi, \bm{x}}} \geq c, \quad \forall e \in \cN.
  \]
\end{definition}

\section{Correlation Gap for Non-Negative Submodular Functions}
\label{app:correlation-gap}

In this section we prove Theorems~\ref{thm:intro-gap} and \ref{thm:intro-gap-mcg2} 
on the correlation gap for non-negative submodular functions. Theorem~\ref{thm:intro-gap}
is a direct approach to the correlation gap, whereas Theorem~\ref{thm:intro-gap-mcg2}
utilizes the Measured Continuous Greedy algorithm. The proofs
of the theorems are qualitatively different and we present them in separate sections.

\subsection{Proof of Theorem~\ref{thm:intro-gap}}

We split the proof into two parts, the upper bound and the lower bound,
state them separately and give their proof.

\paragraph{Upper bound:}
As we remarked in
Section~\ref{sec:intro}, the proof of this upper bound is inspired by the proof in \cite{crs} for the monotone case, which is different from
the earlier one in \cite{vondrak}.

\begin{theorem}\label{thm:upper-gap}
  Let $f : 2^\cN \to \Rp$ be a non-negative submodular function, where
  $n = |\cN|$. Let $\bm{x} \in {[0, 1]}^n$, such that $\bm{x} \leq p \cdot
  \bm{1}_\cN$ for some $p \in [0,1]$. Then,
  \[
    F(\bm{x}) \geq (1-p) \left(1 - \frac{1}{e}\right) f^+(\bm{x}).
  \]
\end{theorem}

\begin{proof}
  Consider a basic feasible solution ${(q_j, A_j)}_{j \in [m]}$ to the
  linear program that defines $f^+(\bm{x})$. In other words,
  $f^+(\bm{x}) = \sum_{j = 1}^m {q_j f(A_j)}$, where
  $\sum_{j = 1}^m {q_j} = 1$, $\sum_{j : i \in A_j} {q_j} = x_i$, for
  all $i$, and $q_j \geq 0$ for all $j$. Notice that, since we chose a
  basic feasible solution and the LP that defines $f^+(\bm{x})$ has only
  $n + 1$ constraints, apart from the non-negativity constraints, we
  have $m \leq n + 1$.

  Next, consider the following process to generate a subset of
  elements. For each $j \in [m]$ sample independently each element of
  $A_j$ with probability $q_j$. An element $i \in \cN$ is not
  selected with probability equal to
  $\prod_{j : i \in A_j} {(1 - q_j)}$, thus, $i$ is selected with
  probability equal to $1 - \prod_{j : i \in A_j} {(1 - q_j)}$. Notice
  that we can assume without loss of generality that $q_j \neq 1$ for
  all $j$; if $q_j = 1$ for some $j$ then that implies that
  $x_i = 1$ for every element $i \in A_j$, and $q_{j'} = 0$ for all
  $j' \neq j$, which then leads us to $F(\bm{x}) = f(A_j) = f^+(\bm{x})$.

  However, we want to make each element $i$ to be selected with
  probability exactly equal to $x_i = \sum_{j : i \in A_j} {q_j}$. To
  do this, we simply need to sample again each element $i$ with
  probability $r_i$, where
  \begin{align}\label{eq:rs}
    1 - (1 - r_i) \cdot \prod_{j : i \in A_j} {(1 - q_j)} &= \sum_{j : i \in A_j} {q_j}.
  \end{align}
  It is  easy to see that $0 \leq r_i \le x_i \le p$.

  Consider the sampling scheme described above, and let $R$ denote a
  random set created via this sampling scheme. Notice that in our
  sampling scheme, each element $i$ is chosen independently with
  probability $x_i$, which implies that $\E_R[f(R)] = F(\bm{x})$.

  We consider $m+n$ sets $B_1,B_2,\ldots,B_{m+n}$ where $B_j = A_j$
  for $1 \leq j \leq m$, and $B_{m+i} = \{i\}$ for $1 \leq i \leq
  n$. Let $\cJ$ denote a random subset of $[m+n]$ obtained by
  including each $j \in \{1, 2, \dots, m\}$ independently with
  probability $q_j$ and each $i \in \{m+1,m+2, \dots, m+n\}$
  independently with probability $r_i$. Also, let $R' \subseteq \cN$
  denote a random set where
  \[
    R' = \bigcup_{j \in \cJ} B_j.
  \]

  The next claim is based on the submodularity of $f$.
  \begin{claim}\label{clm:sampl-submod}
    \[
      F(\bm{x}) \geq \E_{\cJ}\left[f(R')\right].
    \]
  \end{claim}
  \begin{proof}
    Since $F(\bm{x}) = \E_R[f(R)]$, it suffices to show that
    \[
      \E_{R}\left[f(R)\right] \geq \E_{\cJ}\left[f(R')\right].
    \]
    We apply Lemma \ref{lem:sampling} for $k = m+n$, $A_j = B_j$ for
    $1 \leq j \leq m+n$, $p_j = q_j$ for $1 \leq j \leq m$, and
    $p_{m+i} = r_i$ for $1 \leq i \leq n$. Notice that
    \[
      \E_{R}\left[f(R)\right] = \E\left[f\left(\bigcup_{i = 1}^k B_i(p_i)\right)\right],
    \]
    while
    \[
      \E_{\cJ}\left[f(R')\right] = \sum_{I \subseteq [k]} {\left(\prod_{j \in I} {p_j} \prod_{j \notin I} {(1 - p_j)} f\left(\bigcup_{j \in I} {B_j} \right)\right)},
    \]
    and thus, by Lemma \ref{lem:sampling}, we get
    \[
      \E_{R}\left[f(R)\right] \geq \E_{\cJ}\left[f(R')\right].
    \]
  \end{proof}

  \begin{claim}\label{clm:conditioning}
    \[
      \E_{\cJ}\left[f(R')\right] \geq (1-p) \left(1 - \frac{1}{e}\right) \cdot f^+(\bm{x}).
    \]
  \end{claim}
  \begin{proof}
    Assume, without loss of generality, that
    $f(A_1) \geq \dots \geq f(A_m)$. We analyze $\E[f(R')]$ by
    conditioning on the minimum index $j$ that belongs to $\cJ$. For $k \in [m]$, let
    \[
      J_k = \set{I \subseteq [m+n] \: \midd \: k \in I \text{ and }
      \ell \notin I, \forall \ell < k}.
    \]
    Furthermore, for $k \in [m]$ define the set function
    $g_k:2^\cN \rightarrow \Rp$ where $g_k(S) = f(B_k \cup S)$
    for all $S \subseteq \cN$. It is easy to verify that $g_k$ is
    non-negative and submodular because $f$ is non-negative and
    submodular. $\cJ \in J_k$ implies that $B_k \subseteq R'$, hence,
    \begin{align*}
    \E_{\cJ}[f(R') \mid \cJ \in J_k] &= \E_{\cJ}\brk{f(B_k \cup (R' \setminus B_k)) \midd \cJ \in J_k} \\
    &= \E_{\cJ}[g_k(R' \setminus B_k) \mid \cJ  \in J_k].
    \end{align*}

For any fixed $i \in \cN$ we analyze the probability that $i \in R'
\setminus B_k$ conditioned on $\cJ \in J_k$. Using independence of the
choice of each index in $\cJ$ we obtain the following.
\begin{align*}
\Pr_{\cJ}[i \in (R' \setminus B_k) \mid \cJ \in J_k] &=  1 - (1 - r_i) \prod_{j : i \in A_j, k < j \le m} {(1 -  q_j)} \\
    & \le  1 - (1 - r_i) \prod_{j : i \in A_j, j \in [m]} (1 -  q_j)  \\
    &= x_i \leq p.
\end{align*}
Thus, applying Lemma \ref{lem:buchbinder} to $g_k$ yields
\[
\E_{\cJ}[g_k(R'\setminus B_k) \mid \cJ \in J_k] \geq (1-p) g_k(\emptyset) = (1-p) f(B_k).
\]
Combining the above,
\begin{equation}\label{eq:cond-prob}
\E_{\cJ}\left[f(R') \midd \cJ \in J_k \right] \geq (1-p) \cdot f(B_k).
\end{equation}

Also notice that
\begin{equation}\label{eq:prob}
\Pr_{\cJ}[\cJ \in J_k] = \Pr_{\cJ}[k \in \cJ] \cdot \prod_{j = 1}^{k-1} {\left(1 - \Pr_{\cJ}[j \in \cJ]\right)} = q_k \cdot \prod_{j = 1}^{k-1} {(1 - q_j)}.
\end{equation}
Therefore,
\begin{align}\label{eq:cond}
\E_{\cJ}\left[f(R')\right] & = \sum_{k = 1}^m {\Pr_{\cJ}[\cJ \in J_k] \cdot \E_{\cJ}\left[f(R') \midd \cJ \in J_k \right]} \nonumber \\
&\qquad +  \Pr_{\cJ}\left[\cJ \cap [m] = \emptyset \right] \cdot \E_{\cJ}\left[f(R') \midd \cJ \cap [m] = \emptyset \right]  \nonumber \\
& \geq \sum_{j = 1}^m {\Pr_{\cJ}[\cJ \in J_k] \cdot \E_{\cJ}\left[f(R') \midd \cJ \in J_k \right]} \nonumber \\
& \geq \sum_{k = 1}^m {\Pr_{\cJ}[\cJ \in J_k] \cdot (1-p) \cdot f(B_k)} \nonumber \\
& = (1 - p) \sum_{k = 1}^m {q_k f(B_k) \prod_{j = 1}^{k-1} {(1 - q_j)}},
\end{align}
where the first inequality follows from the non-negativity of $f$, the
second inequality follows from \eqref{eq:cond-prob} and the last
equality from \eqref{eq:prob}. However, for $1 \leq k \leq m$, we have
$B_k = A_k$, and thus
\[
\E_{\cJ}\left[f(R')\right] \geq (1 - p) \sum_{k = 1}^m {q_k f(A_k) \prod_{j = 1}^{k-1} {(1 - q_j)}}.
\]

Finally, utilizing Lemma \ref{lem:crs} for $a_k = f(A_k)$ , we get that
\begin{equation}\label{eq:crs}
\sum_{k = 1}^m {q_k  f(A_k) \prod_{j = 1}^{k-1} {(1 - q_j)}} \geq \left(1 - \frac{1}{e}\right) \cdot \sum_{j = 1}^m {q_j f(A_j)} = \left(1 - \frac{1}{e}\right) \cdot f^+(\bm{x}).
\end{equation}

Combining \eqref{eq:cond} and \eqref{eq:crs},
\[
\E_{\cJ}\left[f(R')\right] \geq (1-p) \left(1 - \frac{1}{e}\right) \cdot f^+(\bm{x}).
\]
\end{proof}

Finally, combining Claims \ref{clm:sampl-submod} and \ref{clm:conditioning}, we obtain
\[
  F(\bm{x}) \geq (1-p) \left(1 - \frac{1}{e}\right) \cdot f^+(\bm{x}),
\]
which completes the proof.
\end{proof}

\paragraph{Lower bound:}
A simple example on $n = 2$ shows that $F(\bm{x}) \le (1-p)f^+(\bm{x})$; the
function is the cut function of a directed graph on two vertices. For
monotone functions, a simple coverage example shows that
$F(\bm{x}) \le (1-1/e) f^+(\bm{x})$. We combine and generalize these two examples
to create an instance for non-monotone functions and obtain the following
theorem.

\begin{theorem}\label{thm:lower-gap}
  There exists a non-negative submodular function $f : 2^\cN \to \Rp$ such that, for any $0 \leq p \leq 1$,
  there exists an $\bm{x} \in {[0,1]}^n$ where $\|\bm{x}\|_\infty \leq p$ and
  \[
    F(\bm{x}) \leq \left(1 - e^{-(1-p)}\right) f^+(\bm{x}).
  \]
\end{theorem}
\begin{proof}
  Consider the following graph $G = (V, E)$, where
  $V = \{u_1, \dots, u_n, v\}$, and
  $E = \{(u_i, v) \mid 1 \leq i \leq n\}$. Let
  $x_{u_i} = \frac{1-p}{n}$ for all $i \in \{1, \dots, n\}$ and
  $x_v = p$. We define a function $f : 2^V \to \Rp$ as follows
  \[
    f(S) = \begin{cases}
      1 & \text{if } v \notin S \text{ and } S \neq \emptyset, \\
      0 & \text{otherwise.}
    \end{cases}
  \]

  \begin{figure}[!ht]
    \centering
    \includegraphics[width=0.5\textwidth]{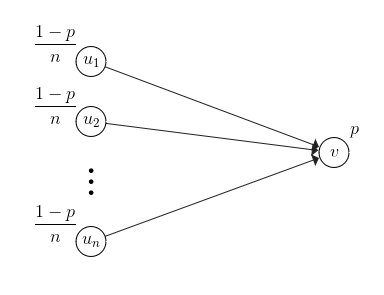}
    \caption{Graph $G$ which yields the desired lower bound.}
  \end{figure}

  It is easy to see that $f$ is submodular. Notice that
  \[
    f^+(\bm{x}) = 1 - p,
  \]
  as the coefficients that maximize $\sum_S {a_S f(S)}$ subject to the
  constraints are $a_{\{v\}} = p$, $a_{\{u_i\}} = \frac{1 - p}{n}$ for
  all $i \in \{1, \dots, n\}$ and $a_S = 0$, for $|S| \neq 1$. In
  other words, $a_{\{u\}} = x_u$ for all $u \in V$, and $a_S = 0$, if
  $|S| \neq 1$.

  Next, notice that, if $R(\bm{x}) \subseteq V$ is a random set, where
  each element $u \in V$ is sampled with probability $x_u$, then
  $f(R(\bm{x})) = 1$ if and only if $v$ is not selected in
  $R(\bm{x})$, but at least one element of $V \setminus \{v\}$ is
  selected. Therefore,
  \[
    F(\bm{x}) = \E[f(R(\bm{x}))] = (1-p) \cdot \left(1 - {\left( 1 - \frac{1-p}{n}\right)}^n\right),
  \]
  which implies that
  \[
    \frac{F(\bm{x})}{f^+(\bm{x})} = \frac{(1-p) \cdot \left(1 - {\left( 1 - \frac{1-p}{n}\right)}^n\right)}{1 - p} = 1 - {\left( 1 - \frac{1-p}{n}\right)}^n.
  \]
  As $n \to \infty$, we get
  \[
    \lim_{n \to \infty} {1 - {\left( 1 - \frac{1-p}{n}\right)}^n} = 1 - e^{-(1-p)}.
  \]

  We conclude that, for any $0 \leq p < 1$, when $x_i \leq p$ for all $i$,
  \[
    F(\bm{x}) \leq \left(1 - e^{-(1-p)}\right) f^+(\bm{x}).
  \]
\end{proof}

\subsection{Proof of Theorem~\ref{thm:intro-gap-mcg2}}

The proof of the correlation gap is via the Measured Continuous Greedy
(\mcg) algorithm and its analysis \cite{mcg}, when applied to an appropriate polytope.
Earlier, we remarked that known results on the \mcg algorithm
\cite{mcg, rocrs} imply that $F_{\max}(\bm{x}) \geq \frac{1}{e} f^+(\bm{x})$.
We quickly sketch the idea implicit in prior work, before we proceed. Let $f:2^\cN \rightarrow \Rp$
be a non-negative submodular function and let $\bm{x} \in [0,1]^n$, where $n = |\cN|$.
Consider a down-closed polytope $\cP$ defined by all points in
$[0,1]^n$ dominated by the given point $\bm{x}$: $\cP \coloneqq \set{\bm{y} \in [0,1]^n \midd \forall \: 1 \leq i \leq n, \: \: y_i \leq x_i}$.
Suppose we run the \mcg algorithm on $\cP$.
From Lemma 8.3 of \cite{rocrs} for $b = 1$, for any $\eps > 0$,
the algorithm can be used to find a point $\bm{z}_\eps \in \cP$ such that
$F(\bm{z}_\eps) \geq \prn{\frac{1}{e} - \eps} \max_{\bm{y} \in \cP} {f^+(\bm{y})} \geq \prn{\frac{1}{e} - \eps} f^+(\bm{x})$.
Since such a point $\bm{z}_\eps \in \cP$ exists for any $\eps >0$, by
the compactness of $\cP$ and the continuity of $F$ and $f^+$, it
follows that there exists a point $\bm{y} \in \cP$ such that
$F(\bm{y}) \geq \frac{1}{e} \cdot f^+(\bm{x})$.
Also notice that $\bm{x} \in \cP$, and thus
\[
F_{\max}(\bm{x}) = \max_{\bm{z} \in \cP} {F(\bm{z})} \geq F(\bm{y}) \geq \frac{1}{e} \cdot f^+(\bm{x}).
\]

To prove Theorem~\ref{thm:intro-gap-mcg2}, we use the same proof
outline as above, but in the algorithm's analysis, we take advantage
of the fact that $\|\bm{x}\|_{\infty} \le p$. Theorem~\ref{thm:intro-gap-mcg2}
generalizes Lemma 8.3 in \cite{rocrs}, used above.

\mcggap*
\begin{proof}
Let $\bm{\hat{x}} = \argmax_{\bm{z} \in \cP} {f^+(\bm{z})}$. Recall that there exists
$\bm{\alpha} \in [0,1]^{2^\cN}$ such that
\[
f^+(\bm{\hat{x}}) = \sum_{S \subseteq \cN} {\alpha_S f(S)}, \quad \sum_{S \subseteq \cN} {\alpha_S} = 1 \: \: \text{ and } \: \: \sum_{S \subseteq \cN} {\alpha_S \1_S} = \bm{\hat{x}}.
\]
From the analysis of Measured Continuous Greedy and the fact that $\bm{x}(b) \in \cP$, we know that,
at time $b$, for all $i \in \cN$ we have
\[
x_i(b) \leq \min\{1 - e^{-b}, p\}.
\]
Let $\bm{x} = \bm{x}(b)$, and, for $S \subseteq \cN$, consider a line of direction
$\bm{d}_S = (\bm{x} \vee \1_S) - \bm{x} = (\1_S - \bm{x}) \vee \bm{0}$. Notice that
$\bm{0} \leq \bm{d}_S \leq \1_S$ for all $S \subseteq \cN$. From Section 2.3 of \cite{ccpv}, it
follows that
\[
\bm{d}_S \cdot \nabla F(\bm{x}) \geq F(\bm{x} \vee \1_S) - F(\bm{x}).
\]
Since $f$ may not be monotone, $\nabla F(\bm{x})$ may have negative entries. Let $\bm{d'}_S$
be a vector obtained from $\bm{d}_S$ as follows: $\prn{\bm{d'}_S}_i = \prn{\bm{d}_S}_i$ if ${\nabla F(\bm{x})}_i \geq 0$,
otherwise $\prn{\bm{d'}_S}_i = 0$. We have $\bm{0} \leq \bm{d'}_S \leq \bm{d}_S$ and,
\[
\bm{d'}_S \cdot \nabla F(\bm{x}) \geq \max\{0, \bm{d}_S \cdot \nabla F(\bm{x})\} \geq \max\{0, F(\bm{x} \vee \1_S) - F(\bm{x})\}.
\]
Since ${\bm{x}(b)}_i \leq \min\{1 - e^{-b}, p\}$ for all $i \in \cN$, by Lemma III.5 of \cite{mcg}, we have
\[
F(\bm{x} \vee \1_S) \geq \prn{1 - \min\{1 - e^{-b}, p\}} f(S).
\]
Therefore,
\begin{align*}
\bm{d'}_S \cdot \nabla F(\bm{x}) &\geq \max\{0, (1-p) f(S) - F(\bm{x}), e^{-b} f(S) - F(\bm{x})\} \\
&\geq \max\{1-p, e^{-b}\} f(S) - F(\bm{x}).
\end{align*}
Next, let $\bm{\hat{d}} = \sum_{S \subseteq \cN} {\alpha_S \bm{d'}_S}$. Since $\bm{d}_S \leq \1_S$ and
$\bm{d'}_S \leq \bm{d}_S$, we have $\bm{d'}_S \leq \1_S$, and thus
\[
\bm{\hat{d}} = \sum_{S \subseteq \cN} {\alpha_S \bm{d'}_S} \leq \sum_{S \subseteq \cN} {\alpha_S \1_S} = \bm{\hat{x}}.
\]
Since $\cP$ is downward-closed and $\bm{\hat{x}} \in \cP$, we know that $\bm{\hat{d}} \in \cP$. Therefore, from the above
and the fact that $\bm{v}_{\text{max}} = \argmax_{\bm{v} \in \cP} {\bm{v} \cdot \nabla F(\bm{x})}$, we
have
\begin{align*}
\der{F(\bm{x}(b))}{b} &= \bm{v}_{\text{max}}(\bm{x}) \cdot \nabla F(\bm{x}) \\
&\geq \bm{\hat{d}}_S \cdot \nabla F(\bm{x}) \\
&= \sum_{S \subseteq \cN} {\alpha_S \cdot \bm{d'}_S \cdot \nabla F(\bm{x})} \\
&\geq \sum_{S \subseteq \cN} {\alpha_S \prn{\max\{1-p, e^{-b}\} f(S) - F(\bm{x})}} \\
&\geq \max\{1-p, e^{-b}\} \sum_{S \subseteq \cN} {\alpha_S f(S)} - \sum_{S \subseteq \cN} {\alpha_S F(\bm{x})} \\
&\geq \max\{1-p, e^{-b}\} f^+(\bm{\hat{x}}) - F(\bm{x}).
\end{align*}
We proceed to solve the above differential inequality. For brevity, let $y = F(\bm{x})$. Then,
\begin{gather}
\dif{y} + y \dif{b} \geq f^+(\bm{\hat{x}}) \max\{1-p, e^{-b}\} \dif{b} \nonumber \\
e^b \dif{y} + y e^b \dif{b} \geq f^+(\bm{\hat{x}}) \max\{(1-p) e^b, 1\} \dif{b} \nonumber \\
\dif{\prn{y e^b}} \geq f^+(\bm{\hat{x}}) \max\{(1-p) e^b, 1\} \dif{b} \nonumber \\
\label{eq:refined-1} y \geq e^{-b} f^+(\bm{\hat{x}}) \int_0^b {\max\{(1-p) e^u, 1\} \dif{u}}.
\end{gather}
Notice that, for $0 \leq u \leq \ln{\prn{\frac{1}{1-p}}}$, we have $(1-p) e^u \leq 1$, while
for $\ln{\prn{\frac{1}{1-p}}} \leq u \leq 1$, we have $1 \leq (1-p) e^u$. Therefore, for
$b \leq \ln{\prn{\frac{1}{1-p}}}$, \eqref{eq:refined-1} becomes
\[
y \geq e^{-b} f^+(\bm{\hat{x}}) \int_0^b {1 \dif{u}} = b \cdot e^{-b} \cdot f^+(\bm{\hat{x}}),
\]
whereas for $b \geq \ln{\prn{\frac{1}{1-p}}}$, \eqref{eq:refined-1} becomes
\begin{align*}
y &\geq e^{-b} f^+(\bm{\hat{x}}) \prn{\int_0^{\ln{\prn{\frac{1}{1-p}}}} {1 \dif{u}} + \int_{\ln{\prn{\frac{1}{1-p}}}}^{b} {(1-p) e^u \dif{u}}} \\
&= \prn{1 - p - e^{-b}\prn{1 + \ln{(1-p)}}} f^+(\bm{\hat{x}}).
\end{align*}
We conclude that
\[
F(\bm{x}(b)) \geq \begin{cases}
b \cdot e^{-b} \cdot \max_{\bm{z} \in \cP} {f^+(\bm{z})},                               & \text{for } 0 \leq b \leq \ln{\prn{\frac{1}{1-p}}} \\
\prn{1 - p - e^{-b} \prn{1 + \ln{(1-p)}}} \cdot \max_{\bm{z} \in \cP} {f^+(\bm{z})},    & \text{for } \ln{\prn{\frac{1}{1-p}}} \leq b \leq 1.
\end{cases}
\]
\end{proof}

\section{Reduction to Small Probabilities}\label{app:reduction}

\begin{observation}[Observation~\ref{obs:reduction}]
  Let $I = (\cN, \cU, \bm{D}, \bm{X}, f, \cC)$ be an instance of the
  Submodular Prophet Inequality problem. For every fixed $\eps > 0$,
  there is a reduction of $I$ to another instance
  $I' = (\cN, \cU', \bm{D'}, \bm{Y}, g, \cC)$ of the SPI problem such
  that that (i) for all $e \in \cU'$, $\bm{D'}(e) \le \eps$ and (ii)
  there exists an $\alpha$-competitive algorithm for $I$ if and only
  if there exists an $\alpha$-competitive algorithm for $I'$.
\end{observation}
\begin{sketch}
Consider the original instance $I$ and recall that each
$\cD_i$ is a probability distribution over $\cU_i$. Our goal is to
ensure that $\cD_i(e) \le \eps$ for every $e \in \cU_i$. Suppose there
is an element $e$ such that $\cD_i(e) > \eps$.  We obtain a new
instance $I'$ as follows.  We replace $e \in \cU_i$ by
$h = \ceil{1/\eps}$ ``copies'' $e_1,e_2,\ldots,e_h$; let $S_e$ denote
this set of copies. Let $\cU'_i$ be the new set of elements. We obtain
a probability distribution $\cD'_i: \cU'_i \rightarrow [0,1]$ as
follows. If $e' \in \cU_i$ such that $e' \neq e$ then
$\cD'_i(e') = \cD_i(e')$ (nothing changes for $e'$). For each copy
$e_j$ of $e$ we set $\cD'_i(e_j) = \cD_i(e)/h$ and by our choice of
$h$ we have $\cD'_i(e_j) \le 1/h \le \eps$, for all $e_j \in
S_e$. Thus, $\sum_{j=1}^h \cD'_i(e_j) = \cD_i(e)$. Since we replaced
$e$ by $h$ copies of it, the ground set $\cU$ changes to $\cU'$ and we
now define a new submodular function $g:\cU' \rightarrow \mathbb{R}_+$
that is derived from $f$. The function $g$ treats the copies of $e$ as
a ``single'' element and hence mimics $f$.  More formally, for any
$A \subseteq \cU'$: $g(A) = f(A)$ if $A \cap S_e = \emptyset$, else
$g(A) = f((A \setminus S_e) \cup \{e\})$. It is easy to verify that if
$f$ is non-negative and submodular, then $g$ is also non-negative and
submodular, and also inherits monotonicity from $f$. Let $I'$ be the
resulting modified instance. We observe that in $I'$, the probability
of an element from $S_e$ being chosen is precisely equal to $\cD_i(e)$
and hence the copies of $e$ act as proxies for $e$ and the submodular
function $g$ ensures that every copy behaves the same as $e$ in $f$.
Note that we crucially relied on the power of submodularity in this
reduction.  One can apply this reduction repeatedly to reduce all
realization probabilities to at most $\eps$. One also notices that the
reduction is computationally efficient as a function of $\eps$. For
any fixed $\eps$, the size of $I'$ is at most $O(1/\eps)$ times the
size of $I$ and a value oracle for $f$ can be used to efficiently and
easily obtain a value oracle for the new submodular function $g$.
\end{sketch}

\section{Submodular Prophet Secretary}\label{app:sps}

A natural question for the prophet inequality setting is whether one
can obtain better prophet inequalities when the arrival order of the
random variables is not chosen by the adversary but instead is chosen
uniformly at random or even chosen by the algorithm. In \cite{esf-prophsec},
the authors introduce the {\em prophet secretary} model, combining the
best of both the secretary and prophet inequality worlds. In particular,
in the prophet secretary model, the arrival order of the random variables
is chosen uniformly at random. There has been much work on this model
and we refer to \cite{Correa-survey,EhsaniHKS18} for several interesting results
in this and related models.

We can consider the {\em Submodular Prophet Secretary (SPS)} problem
as a generalization of the standard prophet secretary problem. The
setting of the \sps problem is exactly the same as the setting of the
\spi problem, with the only difference being that the arrival order of
the random variables in the \sps problem is chosen uniformly at random
instead of adversarially. We note that our results use the OCRS for
the given feasibility constraint in a black-box fashion, and thus we
can use better contention resolution schemes to obtain stronger bounds
for the random-order setting. Specifically, in \cite{rocrs}, the
authors introduce the notion of a {\em Random Order Contention
Resolution Scheme (ROCRS)}, which has improved guarantees compared
to the adversarial order setting. We can use an ROCRS as a black-box,
instead of an OCRS, to obtain improved bounds for the \sps problem.
We do not present these bounds in this version of the paper.

\end{document}